\newif\ifarxiv
\arxivtrue 

\ifarxiv
	\documentclass{acmart}
\else
	\documentclass[format=acmtog]{acmart}
\fi

\usepackage{enumerate}
\usepackage{mathtools}
\usepackage{xcolor,colortbl}
\usepackage{xfrac}
\usepackage{tabu}
\usepackage{enumerate}
\usepackage{stmaryrd}
\usepackage{hyperref}
\usepackage{ulem}
\usepackage{subcaption}
\usepackage{algorithm,algorithmicx}
\usepackage{tikz}
\usetikzlibrary{arrows,positioning,calc,automata,shapes.geometric}
\usetikzlibrary{decorations.pathreplacing,angles,quotes}
\usepackage{bbm}
\usepackage{xspace}
\usepackage{siunitx}

\usepackage[noend]{algpseudocode}
\usepackage[b]{esvect}

 \newtheorem{assumption}{Assumption}
 \newtheorem{problem}{Problem}

\newcommand{\ExpMon}{\texttt{ExpEstimator}\xspace}

\newcommand{\set}[1]{\lbrace #1\rbrace}
\newcommand{\pr}{\mathbb{P}}

\newcommand{\expe}{\mathbb{E}}

\newcommand{\poisson}[1]{\mathit{Poisson}(#1)}

 \newcommand{\twodots}{\mathinner {\ldotp \ldotp}}

\newcommand{\seq}[1]{\vv{#1}}

\algnewcommand{\IfThenElse}[3]{
  \State \algorithmicif\ #1\ \algorithmicthen\ #2\ \algorithmicelse\ #3}
  
\newcommand{\Naturals}{\mathbb{N}}

\title{Runtime Monitoring of Dynamic Fairness Properties}

\author{Thomas A.\ Henzinger}
\affiliation{%
	\institution{Institute of Science and Technology Austria (ISTA)}
	\city{Klosterneuburg}
	\country{Austria}
}
\email{tah@ist.ac.at}

\author{Mahyar	Karimi}
\email{	Mahyar.Karimi@ist.ac.at	}
\affiliation{%
	\institution{Institute of Science and Technology Austria (ISTA)}
	\city{Klosterneuburg}
	\country{Austria}
}

\author{Konstantin Kueffner}
\email{	konstantin.kueffner@ist.ac.at}
\affiliation{%
	\institution{Institute of Science and Technology Austria (ISTA)}
	\city{Klosterneuburg}
	\country{Austria}
}

\author{Kaushik Mallik}
\email{	kaushik.mallik@ist.ac.at	}
\affiliation{%
	\institution{Institute of Science and Technology Austria (ISTA)}
	\city{Klosterneuburg}
	\country{Austria}
}

\copyrightyear{2023}
\acmYear{2023}
\setcopyright{rightsretained}
\acmConference[FAccT '23]{2023 ACM Conference on Fairness, Accountability, and Transparency}{June 12--15, 2023}{Chicago, IL, USA}
\acmBooktitle{2023 ACM Conference on Fairness, Accountability, and Transparency (FAccT '23), June 12--15, 2023, Chicago, IL, USA}
\acmDOI{10.1145/3593013.3594028}
\acmISBN{979-8-4007-0192-4/23/06}



\begin{abstract}
	A machine-learned system that is fair in static decision-making tasks may have biased societal impacts in the long-run. This may happen when the system interacts with humans and feedback patterns emerge, reinforcing old biases in the system and creating new biases. While existing works try to identify and mitigate long-run biases through smart system design, we introduce techniques for monitoring fairness in real time. Our goal is to build and deploy a monitor that will continuously observe a long sequence of events generated by the system in the wild, and will output, with each event, a verdict on how fair the system is at the current point in time. The advantages of monitoring are two-fold. Firstly, fairness is evaluated at run-time, which is important because unfair behaviors may not be eliminated a priori, at design-time, due to partial knowledge about the system and the environment, as well as uncertainties and dynamic changes in the system and the environment, such as the unpredictability of human behavior. Secondly, monitors are by design oblivious to how the monitored system is constructed, which makes them suitable to be used as trusted third-party fairness watchdogs. They function as computationally lightweight statistical estimators, and their correctness proofs rely on the rigorous analysis of the stochastic process that models the assumptions about the underlying dynamics of the system. We show, both in theory and experiments, how monitors can warn us (1) if a bank’s credit policy over time has created an unfair distribution of credit scores among the population, and (2) if a resource allocator’s allocation policy over time has made unfair allocations. Our experiments demonstrate that the monitors introduce very low overhead. We believe that runtime monitoring is an important and mathematically rigorous new addition to the fairness toolbox.
\end{abstract}
	\ifarxiv
		\renewcommand\footnotetextcopyrightpermission[1]{} 
	\else
	\fi

\begin{document}
	\tikzset{elliptic state/.style={draw,ellipse}}
	\maketitle

\section{Introduction}

A majority of works in the fairness literature have considered fairness in static decision making problems, such as classification, regression, etc \cite{feldman2015certifying,hardt2016equality,dwork2012fairness}.
Recent results suggest that fairness itself is not static, but rather dynamic: 
a system that is fair in its static decision-making tasks may become biased in its overall societal impacts  over time \cite{liu2018delayed,heidari2019long,liu2020disparate,mouzannar2019fair,zhang2019group,d2020fairness,zhang2021fairness}.
This happens when the system makes sequential decisions about humans, and every decision of the system is met with some human reaction in return, possibly changing the parameters and the future decisions of the system. 
Such feedback patterns often reinforce historical biases in the dataset and introduce new biases in the society in the long-run as well.
While there are many works that have proposed analysis and mitigation techniques for long-run biases, to our best knowledge, there does not exist any technique that could detect such biases in real-time.
We propose  \emph{runtime monitoring}, as a new addition to the fairness toolbox, for the real-time detection of dynamic social biases in deployed machine-learned decision makers, whose models are unknown and may change over time (e.g., due to retraining, changes in parameters, etc.).

The goal of runtime monitoring is to design a \emph{monitor} which will observe the sequential interactions between the decision-maker and its environment, and, after each observation, will output a quantitative, statistically rigorous estimate of how fair or biased the system is at that point in time.
Unlike  most existing approaches \cite{liu2020disparate,mouzannar2019fair,zhang2019group}, our monitors do not require any assumption or explicit knowledge of the system model.

Monitoring can help us in two ways.
Firstly, by detecting biases in real-time, it can trigger corrective measures or retraining, whenever necessary.
Statically designed fairness interventions are based on an assumed dynamic model of the system.
In practice, models are rarely perfect due to imperfect knowledge of the systems and the involved uncertainties, making it often impossible to predict if a long-run fairness intervention is going to work in practice.
Moreover, the underlying environment conditions may change over time, making static interventions even harder.
Monitoring offers an additional, complementary tool that enables us to close this gap by warning us of the presence of biases in real-time, so that we can adapt our intervention techniques whenever necessary.
There is an analogy to control theory, where it is well-known that closed-loop (feedback) controllers fare much better against modeling uncertainties than open-loop (feed-forward) controllers \cite[Sec.~1.3]{ogata2010modern}. 

The other area where monitoring can help us is in the creation of trusted third-party watchdogs for overseeing the fairness of decision-makers.
They can work neutrally in public interest, since they are by-design independent of the implementation of the system.

Consider the following situation where fairness is dynamic, and we show that monitoring will be useful.
Consider a bank that gives loans to individuals based on their credit scores.
The population is divided into two groups, with one group having higher average credit score than the other.
A policy of the bank that gives loans to the eligible individuals from each group with equal probabilities (equalized opportunity \cite{hardt2016equality}) may seem fair and noble.
However, in doing so, the bank may end up giving more loans to less eligible individuals from the disadvantaged group.
If the credit score distribution of the disadvantaged group is heavily skewed towards a  higher default rate, then there will be many loan defaults, causing a further drop in average credit score of the disadvantaged group \cite{zhang2021fairness}.
For this example, we present a monitor which observes a single long sequence of lending events, consisting of sampling of an applicant, the decision made by the bank on this applicant, and if the loan was granted then whether it was repaid or not.
After each observation, the monitor computes a quantitative statistical estimate of the difference between the average credit scores of the two groups.
It does so by being completely oblivious to the bank's policy and by not assuming any prior knowledge about the humans' behaviors (whether they repay or not).

Now consider the following situation.
It has been shown that voice assistants, such as Amazon Alexa and Google Home, are biased towards the English accents of native speakers, where the native speakers experience significantly higher quality service than the non-native speakers \cite{Harwell2018Alexa}.
This happens when there is an imbalance between group representations in the dataset, with more data available for one demographic group than the other. 
If over time, more and more non-native speakers stop using the service out of dissatisfaction, then the dataset gets more skewed towards the native speakers, intensifying the biases further \cite{hashimoto2018fairness}.
Similar representation-driven biases were reported in other areas as well, such as recommendation systems \cite{chaney2018algorithmic}, credit market \cite{fuster2022predictably}, and crime prediction \cite{ensign2018runaway}.
While, in theory, there are remedies that work if the reactions of the humans can be perfectly predicted, in practice, they may worsen the situation whenever the modeling assumptions do not align well with the true intentions of the humans \cite{zhang2019group}.
This demonstrates that it is difficult to design a static fairness intervention that will always work in the long run. 
Monitoring can help us to, firstly, detect such dynamic biases and warn us in time, and, secondly, to change the interventions whenever necessary.

We consider \emph{time-varying social fairness properties}, as a class of dynamic fairness properties. They can be written as the difference in expected values of a given function over unknown time-varying feature distributions across two demographic groups.
Such properties can capture many existing aspects of long-run fairness properties in the society, such as the time-varying difference in expected credit scores across two groups \cite{zhang2021fairness}, the time-varying difference in group representations \cite{hashimoto2018fairness}, etc.

Our monitors perform statistical estimations to obtain a PAC-style estimate of the value of the social fairness properties in real-time.
We do not make any assumptions about the policies of the already deployed machine-learned agent and the human users (i.e., the environment).
The only assumption we make is that the monitor can observe the features of the selected individual, the actions of the agent, and the reactions of the individual.
Moreover, we assume the availability of a \emph{change function}, such that from each observation the monitor can infer the resulting change in the expected value of the {unknown distribution}.
For instance, in the lending example, we assume the observability of the credit scores and the group memberships of the sampled individuals, the bank's decisions, and the reactions of repaying or defaulting of loans by the individuals.
At any time, if the individual is selected from a group with size $N$, then the change function tells us that a repayment of the loan will increase the credit score of the individual by, say, $1$ point, thereby increasing the average credit score of their group by $\sfrac{1}{N}$.
Similarly, a loan default will decrease the credit score of the individual by, say, $1$ point, thereby decreasing the average credit score of their group by $\sfrac{1}{N}$.
Our monitor observes one long sequence of lending events, and, after each new observation and based on the given history of past lending events and the past valuations of the change function, computes an updated PAC-style estimate of the disparity in average credit scores across the two groups.

Computationally, our monitors are extremely lightweight, and their implementations required only a few lines of code.
Yet, the mathematical analysis of their correctness is nontrivial.
The difficulty stems from the fact that the samples observed on any given sequence are all statistically dependent on each other.
For instance, the probability of sampling an individual with a certain credit score will depend on whether the previous individual who was from the same group and had the same credit score repaid the loan or not. 
As our monitor, we present an unbiased statistical estimator as well as PAC-style bounds for its estimates. 
The bounds are obtained by constructing a martingale from the estimates, analyzing the corresponding martingale difference sequence, and applying suitable concentration inequalities for martingale difference sequences.

We implemented our monitors in a prototype tool.
Using this implementation, we designed monitors for two practical examples from the literature.
The first example concerns the \emph{lending problem} that we discussed earlier, where we monitored, in real-time, to what degree the lending policy of the bank has widened the disparity of average credit scores across the two demographic groups.
The second example is an \emph{attention allocation problem} \cite{d2020fairness}, where incidents keep occurring at every step in multiple locations, and we have a machine-learned allocator for allocating its limited units of attention to the locations to discover the incidents.
The rate at which incidents occur at each location is inversely proportional to the amount of attention allocated to that location in the previous step.
Real-world applications of this example include child services, pest control, etc.
We monitored, in real-time, to what degree the allocator's allocation policy has widened the disparity of discovery probability of incidents among two of the given locations.
Implementations of these systems were already available in the tool \texttt{ml-fairness-gym} \cite{d2020fairness}.
We executed our monitors on the simulation traces of the systems as extracted from \texttt{ml-fairness-gym}. 
We demonstrate that our monitors are able to produce tight statistical estimates of the considered fairness properties in real-time.

We believe that runtime monitors will be an important new addition to the fairness toolbox.
On one hand, they will complement the existing model-based analysis and design tools by checking dynamic fairness in real-time, and helping us to trigger on-demand corrective measures.
On the other hand, they will be useful in building trusted third-party fairness watchdogs.

\subsection*{Related Work}

Fairness in automated decision making has become an active field of research in recent years.
Early works only considered fairness in the static decision making settings, where the decision maker needs to be fair with respect to a time-invariant distribution.
Several group fairness \cite{feldman2015certifying,hardt2016equality} and individual fairness \cite{dwork2012fairness} criteria were proposed, and measures for implementing them were invented.
The proposed measures in this setting can be grouped into three categories: 
(a) ones which \emph{pre-process} the training dataset to eliminate historical biases \cite{calders2013unbiased,kamiran2012data,zemel2013learning,gordaliza2019obtaining},
(b) ones which design training algorithms that are more robust to biases (called \emph{in-processing}) \cite{berk2017convex,zafar2017fairness,zafar2019fairness,agarwal2018reductions}, and
(c) ones which \emph{post-process} the decision-maker's output to eliminate biases \cite{hardt2016equality}.

Later, it was observed by many authors that, surprisingly, decision policies that are statically fair may lead to unfair behavior in the sequential setting.
In this regard, the simplest sequential setting studied in the literature is the two-stage one: in the first stage, the agent makes decisions on humans from two groups, which may cause the humans to take certain actions, and the resulting impact on the groups are then examined in the second stage \cite{liu2018delayed,heidari2019long}.
In the more general long-term setting, the agent is allowed to retrain its decision policy over time, which may be affected by a change or bias in the dataset, caused by the reactions of humans to decisions made by the agent in the past.
This closed feedback was shown to self-reinforce biases that were present in the dataset as well as introduce new biases.
Relevant works on the sequential setting can be found in a recent survey \cite{zhang2021fairness}.
While most of the existing works attempt to eliminate biases at design-time and assume information about the model \cite{liu2020disparate,mouzannar2019fair,zhang2019group}, we detect them at runtime with little knowledge about the model.
There are also simulation-based studies which study long-term impacts of static fairness measures \cite{d2020fairness}.
They are also incomparable to our monitoring setup:
in simulations, it is shown how bias changes over time for an assumed model of interactions, whereas we make almost no assumptions on the model and use the concrete measurements to estimate the bias in the system.

Our monitors are designed to operate in a dynamic setting. Hence, static systems or systems where the decisions of the agent do not affect the parameters of the underlying population, which have been studied extensively in the literature (see Mehrabi et.al.~ \cite{mehrabi2021survey}), are a special case of our setting. Therefore, monitoring could be applied. A natural setting would be the deployment of monitors to check whether an agent in a bandit setting is fair \cite{hossain2021fair,chen2020fair}.

Runtime monitoring is a well-studied subject in the area of formal methods in computer science \cite{bartocci2018lectures}.
The goal is to check, at runtime, if an unknown system satisfies or violates a given safety property.
For instance, a monitor may be used to detect traffic congestion in the roads of a smart city \cite{ma2017runtime}, or safety violation of autonomous vehicles \cite{mao2012runtime}.
The outputs of monitors are usually passed to a safety-supervisory control layer, which takes necessary actions to prevent damages, for example through a default fail-safe action \cite{bloem2015shield}.

Unfortunately, a majority of the existing works in runtime monitoring cannot handle statistical properties, such as fairness.
Notable exceptions include the work by Ferr{\`e}re et al.~\cite{ferrere2019monitoring}, which develops efficient techniques for monitoring statistical properties of systems.
However, they do not consider fairness properties.
Moreover, their monitors' outputs are correct only asymptotically, whereas our monitors output PAC-style error bounds for every observed sequence of finite length.

The closest to our work are the papers by Albarghouthi et al.~\cite{albarghouthi2019fairness} and a recent paper by us~\cite{henzinger2023monitoring}.
Albarghouthi et al.~\cite{albarghouthi2019fairness} presented an approach for monitoring fairness in sequential decision-making tasks, which we generalized to monitoring fairness over Markov chains using techniques from both frequentist and Bayesian statistics~\cite{henzinger2023monitoring}.
These works can be used to monitor only group fairness and individual fairness properties in static decision-making problems, whereas we monitor time-varying social fairness properties in dynamic decision-making problems.

There is a body of research that is ideologically similar to ours, and developed sequential statistical tests to evaluate the performance of already deployed machine-learned systems at runtime.
Podkopaev et al.~\cite{podkopaev2021tracking} proposed an algorithm for monitoring the expected loss of a given classifier due to shift in the dataset distribution.
The expected loss is based on the misclassification rate of a classifier, and is uncomparable to fairness properties that we consider in this work.
Waudby-Smith et al.~\cite{waudby2021time} proposed a sequential estimation algorithm that was used to estimate the time-varying average treatment effect (ATE) in a randomized experiment, which gives a measure of the expected difference in outcome between an individual chosen from the population receiving a treatment (like a medical drug that is being tested) and not receiving the treatment.
Although there are some structural similarities between ATEs and fairness properties, they estimate how much the ATE was \emph{on an average until} the present time, whereas we estimate how much fair the system is \emph{at} the present time.
Moreover, their estimates are asymptotically correct, whereas we provide finite-sample correctness guarantees.


\section{Problem Setup}
\label{sec:general problem}

\subsection{The Sequential Agent-Environment-Interaction Model}

We call a machine-learned decision maker an \emph{agent}, and the population of the subjects of its decisions the \emph{environment}.
For example, in a lending scenario, a bank's machine-learned lending policy is the agent, and the population of the loan applicants is the environment.
We use a setup similar to the work of D’Amour et al.\ \cite{d2020fairness}, where an agent engages in a sequential interaction with its environment, and as a result the parameters of the environment change.
The environment contains a distribution over the individuals, where each individual is represented by a real-valued (scalar) feature of interest, such as their credit score, and a sensitive attribute, such as their ethnicity.  
In this work we only consider fairness properties that depend on the single available feature of the individuals; extension to fairness with respect to feature vectors is left open for future work.
In general, we allow the individuals to have additional features, though they do not influence the fairness.
For simplicity of notation, we suppress such additional unimportant features when considering the individuals.

At each step $t$, the environment samples a \emph{single} individual with feature $X_t$ and the group membership $G_t$, where $X_t$ is a real-valued random variable and $G_t$ is a random variable which is assumed to have a binary support $\set{A,B}$ for simplicity.
We use the shorthand notations $\pr_A(X_t=x_t)$ and $\pr_B(X_t=x_t)$ to denote, respectively,
 the conditional probabilities $\pr(X_t=x_t \mid G_t = A)$ and $\pr(X_t=x_t\mid G_t=B)$.
Moreover, for any two random variables $W$ and $V$ with their respective outcomes $w$ and $v$, we use the shorthand notation $\pr(w\mid v)$ and $\expe(W\mid v)$ instead of $\pr(W=w \mid V=v )$ and $\expe(W\mid V=v)$, respectively.

At time $t$, the agent performs an \emph{action} $Y_t$, which is also treated as a random variable.
Given the agent's action, the environment may react by using its own \emph{reactions}, which we denote using the random variable $Z_t$. 
The randomness in $Y_t$ and $Z_t$ capture the modeling uncertainties, such as unknown factors that influence the agent's actions and unpredictability in the environment's reactions.
In the lending example, the agent's (i.e., the bank's) actions are granting or rejecting the loan to the selected individual, whereas the environment's reactions are repaying or defaulting of the loan by the same individual.
Some problems, such as the attention allocation example, do not require environment's reactions.
(Although, in practice, $Z_t$ may lag from $Y_t$, for simplicity, we assume that they happen at the same time step.)
This completes one round of interaction between the agent and the environment, and a sequence contains many such interaction rounds.

The interactions between the agent and the environment form a sequence of (tuples of) random variables, i.e., a stochastic process $\vec{O} = ((G_t,X_t,Y_t,Z_t))_{t>0}$.
For every $t$, the tuple of concrete values that the random variables take is called an \emph{observation}, denoted as $o_t\coloneqq(g_t,x_t,y_t,z_t)$.
The sequence $\vec{o}_{t}=(o_s)_{s\in[1\twodots t]}=((g_s,x_s,y_s,z_s))_{s\in[1\twodots t]}$ is called an \emph{observation sequence}.

In the process $\vec{O}$, the feature distribution $X_t$ is subject to changes over time; $Y_t$ and $Z_t$ may also change, but that is irrelevant for us. 
We assume that the monitor can infer, from the observations, the resulting change in the current expected value of $X_t$ (given the history of observations).
For instance, in the lending scenario, if at any time the selected individual fails to repay the loan, then the credit score of that individual goes down, and so the distribution of credit scores in the population shifts.
We assume that we can infer the shift in expected credit scores from the lending decision of the bank and the event of repayment/default.
We formalize this in the following.

\begin{assumption}\label{assump:change function}
Runtime monitors have access to a function $\Delta$, called the \emph{change function}, which maps every concrete observation
$o_{t}$ to a change in the expected value of $X_t$, such that for each group $g\in \{A,B\}$, for every time $t$, and for every past sequence of observations $\vec{o}_t$, we have: (i) $\expe_g\left(X_{t+1} \middle| \vec{o}_t \right) = \expe_g\left(X_{t} \middle|\vec{o}_{t-1} \right) + \Delta(o_t)$; (ii) $\expe_g(|X_t|\mid \seq{o}_{t-1})<\infty$; (iii) $X_t$, when centered, is a sub-exponential random variable with parameters $(\sigma^2, \nu)$.
\end{assumption}

Assump.~\ref{assump:change function} imposes mild technical restrictions that are fulfilled by many real-world problems, including the lending example and the attention allocation example that we consider here.
Whenever clear from the context, for simplicity, we write $\Delta_t$ instead of $\Delta(g_t,x_t,y_t,z_t)$.

\subsection{Time-Varying Social Fairness Properties}
Let $f\colon o_t=(g_t,x_t,y_t,z_t)\mapsto \mathbb{R}$ be a function, called the \emph{well-being function}, which is a measure of the well-being of the individual $(g_t,x_t)$ who was subjected to the agent's action $y_t$ to which they reacted with $z_t$.
In the lending example, $f$ maps an observation to the credit score of the selected individual.
In the attention allocation example, $f$ maps an observation to the ratio of the attention (action) to the number of incidents (reaction).

For each group $g$, and for every observation sequence $\vec{o}_{t}$, we define the (group-specific) \emph{expected well-being} as: 
\begin{align}	
	 \omega^g(\vec{o}_t)\coloneqq \expe_g(f(O_t)\mid \vec{o}_{t-1}, y_t, z_t). \label{equ:expected well-being}
\end{align}
Observe that the expectation is with respect to the randomness in the feature distribution $X_t$, which makes $O_t$ also random.
We do not condition on the currently observed feature $x_t$, as it would make the expectation trivially equal to $f(o_t)$.
In other words, the expectation in the well-being is only with respect to the past observations of credit scores in the lending example, and is only with respect to the past observations of incidents in the attention allocation example.

We consider a class of fairness properties, which we call the \emph{time-varying social fairness properties},
defined as the difference in expected well-beings of the two groups for a given observation sequence $\vec{o}_t$:
\begin{align}
	\varphi(\vec{o}_{t}) \coloneqq \omega^A(\vec{o}_t) - \omega^B(\vec{o}_t). \label{equ:general fairness spec}
\end{align}
Time-varying social fairness properties capture many interesting properties that were already studied in the context of sequential decision-making, such as the time-varying disparity in average credit score \cite{liu2018delayed}, time-varying disparity in the discovery probability of incidents \cite{ensign2018runaway,elzayn2019fair}, etc.

In \eqref{equ:general fairness spec}, we present the general class of time-varying fairness properties that we consider, and the exact property will depend on the application and the definition of the function $f$.
For instance, in the lending example, $f(o_t)$ will be independent of $g_t,y_t,z_t$ and will give us the credit score of the individual sampled at time $t$.

We point out that we do not impose any assumption on the agent's and the environment's policies for choosing their respective actions and reactions.
However, following Assumption~\ref{assump:change function}, (re-)actions at each time influence the expected observation at the next step. 
Hence, it is impossible to statically predict the conditional expectation in advance. 
Intuitively, this means that without observing the loan decisions of the bank and the subsequent repayment or default events, we cannot predict what the expected credit score will be at a particular point in the future.

As a result, we cannot statically predict the social fairness in the system in the long-run, even if we knew its initial value.
Thereby, it is only possible to measure social fairness retrospectively, which is what we do using runtime monitoring.
To our best knowledge, no prior work in the fairness literature considered this problem.

\subsection{The Monitoring Problem}
A \emph{monitor} is a function that maps every observation sequence to a real interval, where the output interval computed by the monitor is a PAC-style statistical estimate of the given social fairness property.
We summarize the monitoring problem in the following.

\begin{problem}\label{prob:general problem}
	Let $\vec{O}$ be a stochastic process, $\varphi$ be a social fairness property, and $\delta\in [0,1]$ be a parameter.
	Design a monitor $M$ such that for every time $t$, the following holds:
	\begin{align*}
		\pr\left( \varphi(\vec{O}_{t})  \in M(\vec{O}_{t}) \right) \geq 1-\delta.
	\end{align*}
\end{problem}

The probabilistic uncertainty in the monitor's output is due to the non-availability of the parameters of the initial feature distribution: 
were the initial parameters known to the monitor, at every time, a precise value of the fairness property could be calculated from the net change in the parameters as deduced from the change function. 
On the other hand, a na\"ive PAC estimate of $\varphi(\vec{O}_{t})$ at each time step is also not feasible, since the feature distribution is constantly changing.

For the fixed observation sequence $\vec{o}_{t}$, the estimate $[l,u]=M(\vec{o}_{t})$ is called the $(1-\delta)\cdot 100\%$ \emph{confidence interval} for $\varphi(o_{t})$.
The radius, given by $\varepsilon=0.5\cdot (u-l)$, is called the \emph{estimation error}, and the quantity $1-\delta$ is called the \emph{confidence}.
The estimate gets more precise as the error gets smaller and the confidence gets higher.
For the lending example, Prob.~\ref{prob:general problem} asks us to design a monitor which will observe a sequence of lending events, and, after each observation, will output a $(1-\delta)\cdot 100\%$ confidence interval for the estimated disparity in average credit scores.

While our monitors output \emph{interval estimates} of fairness properties in the form of confidence intervals, internally, they first compute \emph{point estimates} of the \emph{expected feature} of each group $g$ for a given observation sequence $\seq{o}_{t-1}$, defined as:
\begin{align}
	\psi^g(\seq{o}_{t-1}) \coloneqq \expe_g(X_t\mid \seq{o}_{t-1}). \label{equ:expected feature}
\end{align}
Note that the quantity, $\psi^g(\seq{o}_{t-1})$ gives us only the expected feature, which will be an intermediate step for estimating the well-being.
Notice that the expected feature at time $t$ only depends on the past observations until time $t-1$, whereas the well-being at time $t$ requires the action and the reactions, such as units of attentions allocated by the attention allocator at the current time $t$.
A \emph{point estimator} $\hat{E}$ of $X_t$ for a given $\seq{o}_t$ and a given group $g$ is a function $\hat{E}\colon \vec{o}_t\mapsto \mathbb{R}$.
Additionally, $\hat{E}$ will be called \emph{unbiased}, if for every observation sequence $\vec{o}_{t-1}$ that may occur with positive probability, we have:
\begin{align*}
	 \expe_g\left( \hat{E}(\vec{O}_t) \,\middle|\, \vec{o}_{t-1}\right) = \psi^g(\vec{o}_{t-1}).
\end{align*}
Intuitively, unbiasedness guarantees that, for any given history $\vec{o}_{t-1}$ of loan events that may occur with positive probability, the expected credit score of a group at time $t$ will be equal to the expected output of the estimator at time $t$.

While unbiasedness guarantees that the estimator $\hat{E}$'s output coincides with $\psi^g(\vec{o}_{t-1})$ \emph{in expectation}, we also require that the output error remains statistically bounded at all time.
To this end, we bound the estimation error by computing confidence intervals for $\psi^g(\vec{o}_{t-1})$, obtained through application of concentration inequalities around the point estimate.
These confidence intervals of group-specific expected features are then used to obtain confidence intervals for group-specific expected well-beings (i.e., $\omega^g(\vec{o}_t)$), which are then subtracted from each other to finally obtain the output confidence interval of the monitor for the time-varying social fairness property $\varphi(\vec{o}_t)$. 
We illustrate this sequence of steps in Fig.~\ref{fig:operational diagram}, along with references to the sections where they are described.

While we provide a general procedure for estimating $\psi^g(\vec{o}_{t-1})$ in the first step, a general overall estimation procedure for any arbitrary $\varphi(\vec{o}_t)$ is difficult to derive and is left open.
This is because the final confidence interval for $\varphi(\vec{o}_t)$ depends on the structure of the well-being function, which is problem-specific.

As a convention, by ``monitor,'' we will exclusively refer to the final interval estimator of $\varphi(\vec{o}_t)$, though it is not the only interval estimator that we will use.

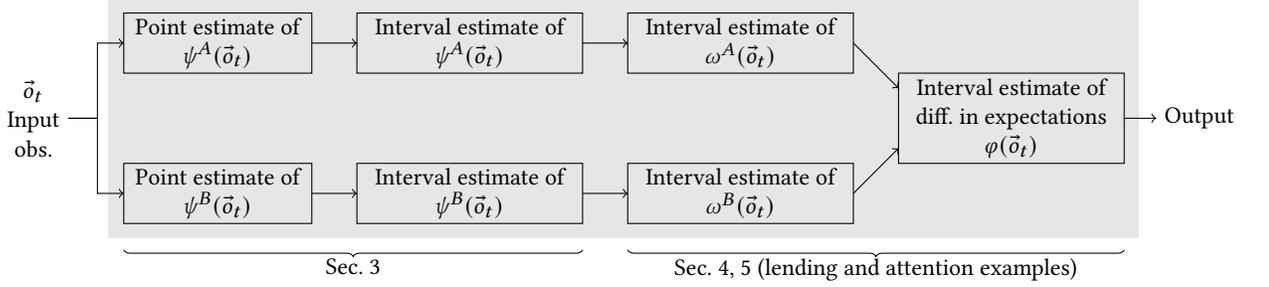
\begin{figure*}
	\begin{tikzpicture}
		\fill [black!10!white]	(-0.2,-2.2)	rectangle	(13.5,1.0);
		
		\draw	(0,0)	rectangle	node	[align=center]	{Point estimate of\\ $\psi^A(\vec{o}_t)$}	(2.5,0.8);
		\draw	(3.1,0)	rectangle	node	[align=center]	{Interval estimate of\\ $\psi^A(\vec{o}_t)$}	(6.1,0.8);
		\draw	(6.7,0)	rectangle	node	[align=center]	{Interval estimate of\\ $\omega^A(\vec{o}_t)$}	(9.7,0.8);
		
		\draw	(0,-2)	rectangle	node	[align=center]	{Point estimate of\\ $\psi^B(\vec{o}_t)$}	(2.5,-1.2);
		\draw	(3.1,-2)	rectangle	node	[align=center]	{Interval estimate of\\ $\psi^B(\vec{o}_t)$}	(6.1,-1.2);
		\draw	(6.7,-2)	rectangle	node	[align=center]	{Interval estimate of\\ $\omega^B(\vec{o}_t)$}	(9.7,-1.2);
		
		\draw	(10.3,-1.2)	rectangle	node	[align=center]	{Interval estimate of\\ diff.\ in expectations \\ $\varphi(\vec{o}_t)$}	(13.3,0);
		
		\draw[->]	(2.5,0.4)	--	(3.1,0.4);
		\draw[->]	(6.1,0.4)	--	(6.7,0.4);
		\draw[->]	(9.7,0.4)	--	(10.3,-0.2);
		\draw[->]	(2.5,-1.6)	--	(3.1,-1.6);
		\draw[->]	(6.1,-1.6)	--	(6.7,-1.6);
		\draw[->]	(9.7,-1.6)	--	(10.3,-1);
		
		\node[align=center]	(A)	at	(-1.2,-0.6)	{$\vec{o}_t$\\ Input\\ obs.};

		\draw[->]	(A.east)	--	(-0.35,-0.6)	--	(-0.35,0.4)	--	(0,0.4);
		\draw[->]	(A.east)		--	(-0.35,-0.6)	--	(-0.35,-1.6)	--	(0,-1.6);
		
		\node	(B)	at	(14.3,-0.6)	{Output};
		\draw[->]	(13.3,-0.6)	--	(B.west);
		
		\draw[decoration={brace,mirror,raise=10pt},decorate]	(0,-2)		--	node	[below=10pt]		{Sec.~\ref{sec:interval estimator expected value}}	(6.1,-2);
		\draw[decoration={brace,mirror,raise=10pt},decorate]	(6.7,-2)		--	node	[below=10pt]		{Sec.~\ref{sec:lending}, \ref{sec:attention} (lending and attention examples)}	(13.3,-2);
	\end{tikzpicture}
	\caption{The operational diagram of the monitor and the sections of this paper where the components are presented.
	}
	\label{fig:operational diagram}
\end{figure*}


\section{An Interval Estimator for the Time-Varying Expected Feature}
\label{sec:interval estimator expected value}

For any group $g$ and an arbitrary observation $\vec{o}_t$, 
in this section, we construct an interval estimator for the group-specific expected feature $\psi^g(\vec{o}_{t-1})=\expe_g(X_t\mid \seq{o}_{t-1})$ (defined in \eqref{equ:expected feature}).
As $\psi^g(\vec{o}_{t-1})$ does not depend on the fairness metric, hence the estimator is not tied to the fairness monitoring problem and can have other use. 
The estimator for $\psi^g(\vec{o}_{t-1})$ will be the essential component of the monitors (i.e., the interval estimator for the social fairness property $\varphi(\vec{o}_t)$), which will be presented later in the respective example sections.
The interval estimate of $\psi^g(\vec{o}_{t-1})$ is obtained by first computing a point estimate of it, and then using concentration inequalities to bound the estimation error.
The first part, i.e., the point estimation part, is explained using a coin-toss analogy.

\subsection{Warm-Up: A Coin-Toss Puzzle}

Suppose we have a coin with unknown initial bias.
After each toss, its bias changes in a predefined manner as a function of the outcome of the previous toss.
How to compute a point estimate of the bias of the coin at any given point in time, based on the given sequence of the observed past outcomes?

Let, us formalize the problem first.
Let, at time $t$, the probability that the coin shows $1$ (heads) be $p_t$,  and the probability that the coin  shows $0$ (tails) be $1-p_t$.
The toss outcome at time $t$ is denoted using the random variable $X_t$.
Let $p_1$ be the initial bias which is fixed but unknown.
After every toss, the coin's bias shifts by a constant $\epsilon\in [0,1]$, and the direction of the shift depends on the outcome of the previous toss:
if we see $X_t=1$ at time $t$, the bias shifts to $p_{t+1} = p_t+\epsilon$ in the next step, whereas if we see $X_t=0$ at time $t$, the bias shifts to $p_{t+1}=p_t-\epsilon$.
Let's assume for simplicity that the true initial bias is not too close to the boundaries $0$ and $1$, and moreover $\epsilon$ is small enough and the observed sequence is short enough that the true bias never reaches the boundaries.
We can succinctly write $p_{t+1} = p_t + \Delta_t$, where $\Delta_t = x_t\epsilon + (1-x_t)(-\epsilon)$ is the change of bias recorded at time $t$.
Then the estimation problem asks: for any given sequence $x_1,\ldots,x_t$ of toss outcomes, how to compute a point estimate of the unknown parameter $p_t$, i.e., the value of 
$\expe(X_t \mid x_1, \dots, x_{t-1})$?

Observe that even if we knew $p_1$, without seeing the observations $x_1,\ldots,x_t$, the value of $p_t$ would only be probabilistically known.
Thus a static analysis would not be possible even in that case.

For the trivial case of $\epsilon=0$, i.e., when we know that the bias remains fixed at $p_1$ throughout, we can compute an unbiased point estimate of $p_t=p_1$ by simply computing the empirical average of the observed sequence as $\hat{p}_t = \frac{1}{t}\sum_{s=1}^t x_s$. 

When $\epsilon > 0$, we show 
\ifarxiv
	(see App.~\ref{lemma:estimator_unbiased}) 
\fi
that, for the given observation sequence $x_1,\ldots,x_t$, the following is an unbiased point estimator of $p_1$:
\begin{align}
	\hat{p}_1 = \frac{1}{t}\sum_{s=1}^t \left( x_s - \sum_{r=1}^{s-1} \Delta_r \right). \label{equ:estimating time-varying bias in a coin time 1}
\end{align} 
Once an estimate of $p_1$ is known, we can obtain an unbiased point estimate for $\expe(X_t \mid x_1, \dots, x_{t-1})$ by accounting for the observed changes up to time $t$:
\begin{align}
	\hat{p}_t = \frac{1}{t}\sum_{s=1}^t \left( x_s - \sum_{r=1}^{s-1} \Delta_r \right) + \sum_{s=1}^{t-1} \Delta_s. \label{equ:estimating time-varying bias in a coin}
\end{align} 

While there are techniques to estimate the non-time-varying mean of a statistical process from sequential observations \cite{waudby2020estimating}, to our best knowledge, the problem we consider and the solution we propose are completely novel.

\subsection{The Interval Estimator for Expected Features}

Now we extend the coin-toss analogy to the point estimation, followed by the interval estimation, of the expected features $\psi^g(\cdot)$.
For simpler notation, in the rest of this section, we assume that there is only one group, so that all the past observations correspond to that group only.
We drop the superscript $g$ from the property $\psi^g(\cdot)$.

Drawing comparison with the coin-toss setting,
the bias $p_t$ of the coin at every $t$ is now replaced by $\expe(X_t\mid \vec{o}_{t-1})$, and the bias shift $\Delta_t$ of the coin at time $t$ is now replaced by the value of the change function $\Delta(o_{t})$.
We make these adjustments in the point estimator of bias of the coin in \eqref{equ:estimating time-varying bias in a coin}, and obtain the following point estimator for $\psi^g(\vec{o}_{t-1})$:
\begin{align*}
	\hat{E}(\vec{o}_t) \coloneqq \frac{1}{t} \sum_{s=1}^{t} \left( x_s - \sum_{r=1}^{s-1} \Delta(o_r)\right) + \sum_{s=1}^{t-1} \Delta(o_s).
\end{align*}
From this point estimator, we can obtain an interval estimator $\psi(\vec{o}_{t-1})$ by applying a suitable version of Azuma-style inequality (see Theorem \ref{theorem:estimator_concentration_general}) to compute a $(1-\delta)\cdot 100\%$ confidence interval around $\hat{E}(\vec{o}_t)$ for any given $\delta$. 
We call the interval estimator of $\psi(\cdot)$ \ExpMon, and present its pseudocode in Alg.~\ref{alg:general monitor:single random variable}.
In the function $\mathtt{Init}$, the monitor first initializes various internal registers.
After each new observation $(g,x,y,z)$, the monitor invokes the function $\mathtt{Compute}$ to compute a new $(1-\delta)\cdot 100\%$ confidence interval for the expected features $\psi^g(\vec{o}_{t-1})$.

\begin{algorithm*}
	\caption{\ExpMon}
	\label{alg:general monitor:single random variable}
	\begin{minipage}{0.45\textwidth}
		\begin{algorithmic}[1]
			\Function{$\mathtt{Init}$}{$\Delta,\delta,\sigma,\nu$}
				\State $\Delta \gets \Delta$ \Comment{change function}
				\State $\delta \gets \delta$ \Comment{target confidence level}
				\State $(\sigma,\nu)\gets (\sigma,\nu)$ \Comment{sub-exp. parameters}
				\State $t\gets 0$ \Comment{clock counter}
				\State $\hat{E}_1\gets 0$ \Comment{curr.\ estimate of the initial expected value}
				\State $d \gets 0$ \Comment{net distribution shift}
			\EndFunction
		\end{algorithmic}
	\end{minipage}
	\begin{minipage}{0.45\textwidth}
		\begin{algorithmic}[1]
			\Function{$\mathtt{Compute}$}{$x,y$}
				\State $t\gets t+1$
				\State $\hat{E}_1\gets \frac{1}{t}\left(\hat{E}\cdot (t-1) + (x-d)\right)$ \Comment{update}
				\State $\hat{E} \gets \hat{E}_1+d$ \Comment{curr.\ estimate}
				\State $d\gets d+\Delta(x,y)$  \Comment{update net shift}
				\State $\varepsilon \gets \mathit{Azuma}(t,\delta,\sigma,\nu)$ \Comment{see Theorem \ref{theorem:estimator_concentration_general}}
				\State \Return $[\hat{E}-\varepsilon, \hat{E}+\varepsilon]$
			\EndFunction
		\end{algorithmic}
	\end{minipage}
\end{algorithm*}

\subsection{Soundness of the Interval Estimator}

For soundness, we need to show that (a) the point estimator of $\psi(\vec{o}_{t-1})$ is unbiased,  
and that (b) the interval estimate computed using the Azuma-style inequality is statistically sound.
Claim~(a) follows from the definition of unbiasedness. 
For claim~(b), we show that the sequence of the expected point estimates, conditioned on the increasingly longer sequence of prefixes of a given observation, is a Doob martingale.
Furthermore, we show that the difference between any two consecutive elements of the Doob martingale is a sub-exponential random variable, which enables us to use an Azuma-style concentration inequality to compute the desired confidence interval from the point estimate. 
Below, we present the highlights of the soundness proof 
\ifarxiv
	(details in App.~\ref{sec:proofs_general_monitor}),
\fi
 which can be skipped over without any loss of continuity.

\smallskip
\noindent\textbf{Unbiasedness:}
To demonstrate that the estimator $\hat{E}(\vec{O}_t)$ is unbiased, we first show: 

\begin{lemma}
	\label{lemma:main unbiased}
	The estimator $\hat{E}_1(\vec{o}_t)\coloneqq \frac{1}{t} \sum_{s=1}^{t} \left( x_s - \sum_{r=1}^{s-1} \Delta(o_r)\right)$ is an unbiased estimator of $\expe\left(X_1\right)$.
\end{lemma}
Then we utilize the change function $\Delta$ and the definition of conditional expectation to transfer the result to $\hat{E}(\vec{O}_t)$.
\begin{align}
	&\expe\left(X_t \;\middle|\; \vec{O}_{t-1}\right) 
	= 
	\expe(X_1)+ \sum_{s=1}^{t-1} \Delta(O_s) \\
	&= 
	\expe\left(\hat{E}_1(\vec{O}_t)\right) + \sum_{s=1}^{t-1} \Delta(O_s) 
	=
	\expe\left(\hat{E}(\vec{O}_t)\;\middle|\; \vec{O}_{t-1} \right)
\end{align}

\smallskip
\noindent\textbf{Proving concentration around the mean using martingales:}
To demonstrate that the estimator $\hat{E}(\vec{O}_t)$ concentrates around its mean, we construct the \emph{Doob martingale}: $(\expe( \hat{E}_1(\vec{O}_t)\mid \vec{O}_s))_{s\in [0\twodots t]}.$
Intuitively, this martingale is a step-by-step approximation process. That is, it starts from the quantity we want to estimate, i.e.\ $\expe(\hat{E}_1(\vec{O}_t))$ which due to Lemma \ref{lemma:main unbiased} is equal to $\expe(X_1)$, and ends at $\expe( \hat{E}_1(\vec{O}_t) \mid \vec{O}_t)$ which given the definition of the conditional expectation is the complete random variable, i.e.\ our estimator $\hat{E}_1(\vec{O}_t)$.

\smallskip
\noindent\textbf{Applying an Azuma-style concentration inequality:}
To bound the distance between the first and last step of the martingale we want to apply some form of Azuma-style concentration inequality. However, this requires knowledge about the behavior of the difference between two consecutive martingale steps.

\begin{lemma}
	\label{lemma:main martingale difference}
	The martingale difference for any $s\in[1\twodots t]$
	\begin{align*}
		\expe\left( \hat{E}_1(\vec{O}_t)\; \middle| \; \vec{O}_{s+1}\right) - \expe\left( \hat{E}_1(\vec{O}_t)\; \middle| \; \vec{O}_{s}\right)=
		\frac{1}{t}\left(X_{s+1} - \expe\left( X_{s+1}\;  \middle| \; \vec{O}_s\right)\right)
	\end{align*}
	and therefore it is a sub-exponential random variable with parameter $(\frac{\sigma^2}{t^2}, \frac{\nu}{t})$.
\end{lemma}

Hence, we can bound the probability that $\expe(X_1)$ is within an $\varepsilon>0$ interval around the estimator $\hat{E}_1(\vec{O}_{t})$, i.e.\ $\pr(| \expe( \hat{E}_1(\vec{O}_t)) - \expe(X_1)| \geq \varepsilon ) $. This directly translates into a bound of the difference between $\expe( X_t\mid \vec{O}_{t-1}  )$ and $\hat{E}( \vec{O}_{t})$,
\begin{align*}
	&\pr\left(\left| \expe\left( \hat{E}_1(\vec{O}_t)\right) - \expe(X_1) + \sum_{s=1}^{t-1} \Delta(O_s)-\Delta(O_s)\right| \geq \varepsilon \right) \\
	&= \pr\left(\left| \expe\left( \hat{E}(\vec{O}_t)\right)  - \expe\left( X_t\; \middle|\; \vec{O}_{t-1}  \right)\right| \geq \varepsilon \right)
\end{align*}
Thus we finally obtain the soundness theorem of our expected value estimator.

\begin{theorem}
    \label{theorem:estimator_concentration_general}
    Let $\delta\in [0,1]$.
	Let $\vec{O}$ be a stochastic process as specified in Assumption \ref{assump:change function}.
	For every time step $t$,
	\begin{align*}
        \left| \psi(\vec{O}_{t-1})-\hat{E}(\vec{O}_t) \right| \leq  \max\left\lbrace 
		\sqrt{\frac{2\sigma^2}{t} \cdot \log\left( \frac{2}{\delta}\right)} ,
		\frac{2\nu}{t}\log\left( \frac{2}{\delta}\right) \right\rbrace
    \end{align*}
	holds with probability $1-\delta$. We define $\mathit{Azuma}(t,\delta,\sigma^2,\nu)$ to refer to this bound.
\end{theorem}


\section{A Dynamic Lending Problem}
\label{sec:lending}

Now we present a monitoring algorithm for the lending example that we took from the literature \cite{liu2018delayed}.
Suppose we have a bank that lends money to individuals by taking into account their credit score and group membership information.
Every individual has a credit score that may change over time, and let $[0\twodots c_{\mathsf{max}}]$ be the set of all possible credit scores of every individual.
Also, the population is divided into groups $A$ and $B$, as usual, and let the number of individuals in the groups be $N_A$ and $N_B$, respectively. 
At every step, the bank gets the credit score and the group information of a randomly chosen individual, and decides whether to accept or reject the loan application.
If the loan is granted and is subsequently repaid, then the respective individual's credit score increases by $1$, provided her initial credit score was smaller than $c_{\mathsf{max}}$.
On the other hand, if the loan is granted but defaulted, then the credit score decreases by $1$, provided her initial credit score was greater than $0$.
If the loan is rejected, then the individual's credit score remains unchanged.
We want to monitor, after each observation, if the bank's policy leads to unfair distribution of the expected credit score among the individuals of the two groups $A$ and $B$.

\subsection{Problem Formulation}

We assume a uniform distribution over the set of individuals at every time step.
Given a (random) individual with features $X_t$ and group $G_t$ sampled uniformly at time $t$, the bank's (random) action $Y_t$ of rejecting or accepting the individual is given by $Y_t=0$ or $Y_t=1$, respectively.
Once the bank has chosen an action, the individual reacts as follows:
if $Y_t=0$ then the individual's action is immaterial, and if $Y_t=1$ then the individual performs (random) action $Z_t$, i.e., they can pick either $Z_t=1$ or $Z_t = 0$ denoting, respectively, whether they repay the loan or not.
The resulting change in the distribution is specified using the change function $\Delta$ defined as below:
\begin{align*}
	\Delta(x_t,g_t,y_t,z_t) = \begin{cases}
								+\frac{1}{N_{g_t}} & \text{if } c<c_{\mathrm{max}},\, y_t = 1,\, z_t = 1 ,\\
								-\frac{1}{N_{g_t}} & \text{if } c>0,\, y_t = 1,\, z_t = 0,\\
								0				& \text{otherwise.}
							\end{cases}
\end{align*}
That is, if the loan is granted and the individual repays the loan, the expected credit score increases by $\frac{1}{N_{g_t}}$; if the loan is granted and the individual fails to repay the loan, the expected credit score decreases by $\frac{1}{N_{g_t}}$; otherwise the expected credit score remains the same.

The well-being function $f$ maps to the credit score (i.e., the feature itself), and the time-varying social fairness criteria is given by the disparity in expected credit scores of the two groups, i.e., 
\begin{align*}
	\varphi_t = \expe_A(X_{t} \mid \vec{o}_{t-1}) - \expe_B(X_{t}\mid \vec{o}_{t-1}). 
\end{align*}

\subsection{The Runtime Monitor}

The monitor for the lending problem, called \textsf{LendingMonitor}, is shown in Alg.~\ref{alg:lending monitor}.
Internally, \textsf{LendingMonitor} runs two \ExpMon (see Alg.~\ref{alg:general monitor:single random variable}) sub-monitors $M_A$ and $M_B$ in parallel, which compute the interval estimates for the expected credit scores of the two groups respectively.
After each new observation, depending on the group information of the selected individual, the respective monitor's \texttt{Compute} function is invoked, and a new estimate for that group is computed.
This is possible because $X_t$ is a bounded random variable, making it sub-gaussian with parameter $c_{\mathsf{max}}^2$ (which implies it is sub-exponential with $\nu=0$).
The output of \textsf{LendingMonitor} is the difference between the two interval estimates computed by $M_A$ and $M_B$.
Observe that each of $M_A$ and $M_B$ uses higher confidence level $1-\frac{\delta}{2}$, so that the final confidence of the output estimate becomes $1-\delta$, after applying the union bound. 
For simplicity, we chose union bound to compute the overall confidence, however, tighter interval estimates can be computed by observing that the group-specific stochastic processes are statistically independent, thereby allowing us to use the sharper bounds from the Hoeffding's inequality (see \cite[p.~24]{wainwright2019high}).

\begin{algorithm}
	\caption{\textsf{LendingMonitor}}
	\label{alg:lending monitor}
	\begin{minipage}{0.55\textwidth}
		\begin{algorithmic}[1]
			\Function{$\mathtt{Init}$}{$\Delta,\delta,c_{\mathsf{max}}$}
				\State $M_A \gets \ExpMon.\mathtt{Init}(\Delta,\frac{\delta}{2},c_{\mathsf{max}},0)$
				\State $M_B \gets \ExpMon.\mathtt{Init}(\Delta,\frac{\delta}{2},c_{\mathsf{max}},0)$
				\State $\overline{E}_A \gets 0$
				\State $\overline{E}_B \gets 0$
			\EndFunction
		\end{algorithmic}
	\end{minipage}
	\begin{minipage}{0.4\textwidth}
		\begin{algorithmic}[1]
			\Function{$\mathtt{Compute}$}{$x,g,y,z$}
				\State $\overline{E}_g\gets M_g.\mathtt{Compute}(x,g,y,z) $
				\State $\overline{\varphi} \gets \overline{E}_A - \overline{E}_B$ \Comment{interval difference}
				\State \Return $\overline{\varphi}$
			\EndFunction
		\end{algorithmic}
	\end{minipage}
\end{algorithm}

\subsection{Experimental Outcome}

\begin{figure*}
	\includegraphics[scale=0.3]{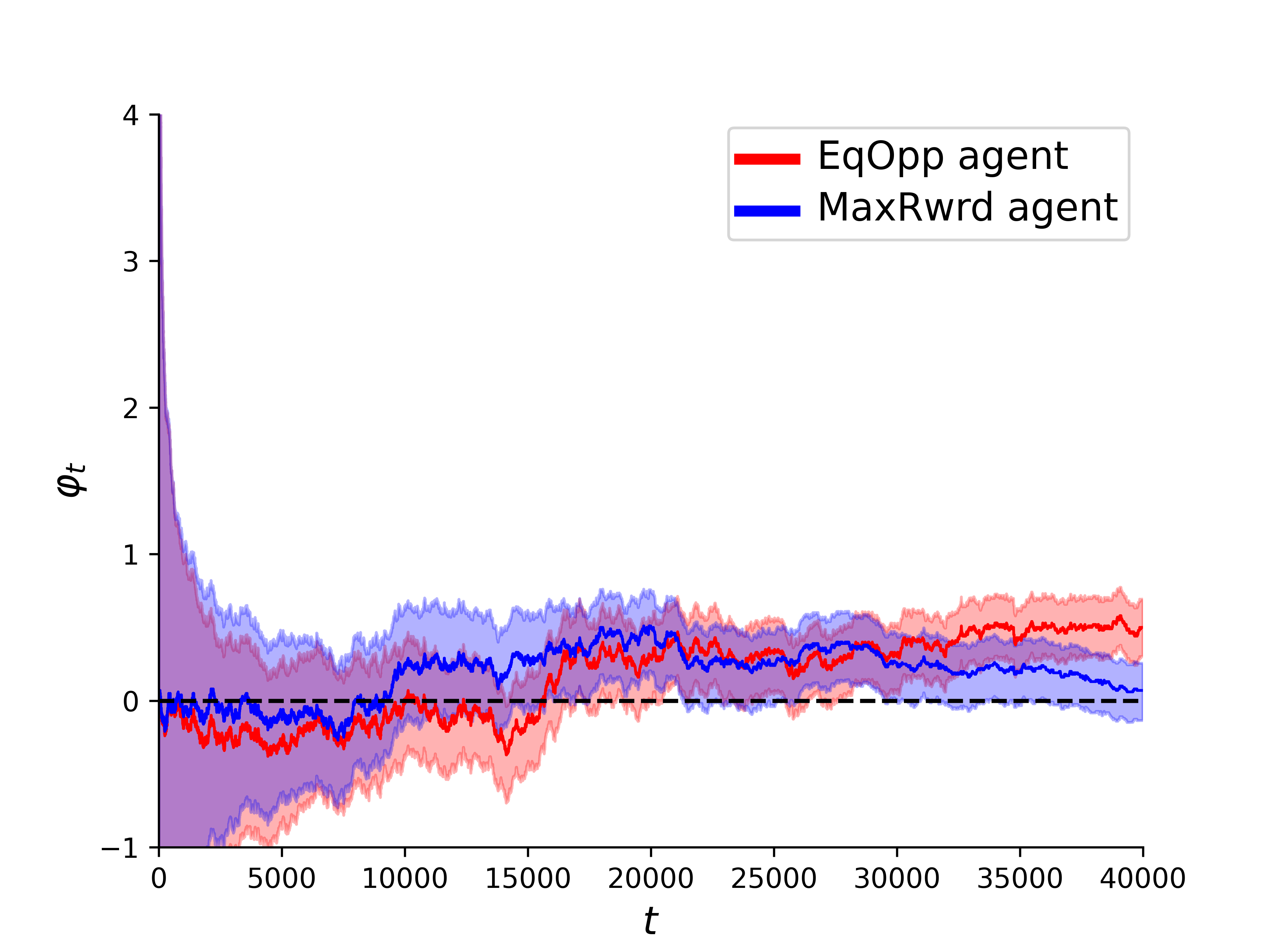}
	\includegraphics[scale=0.3]{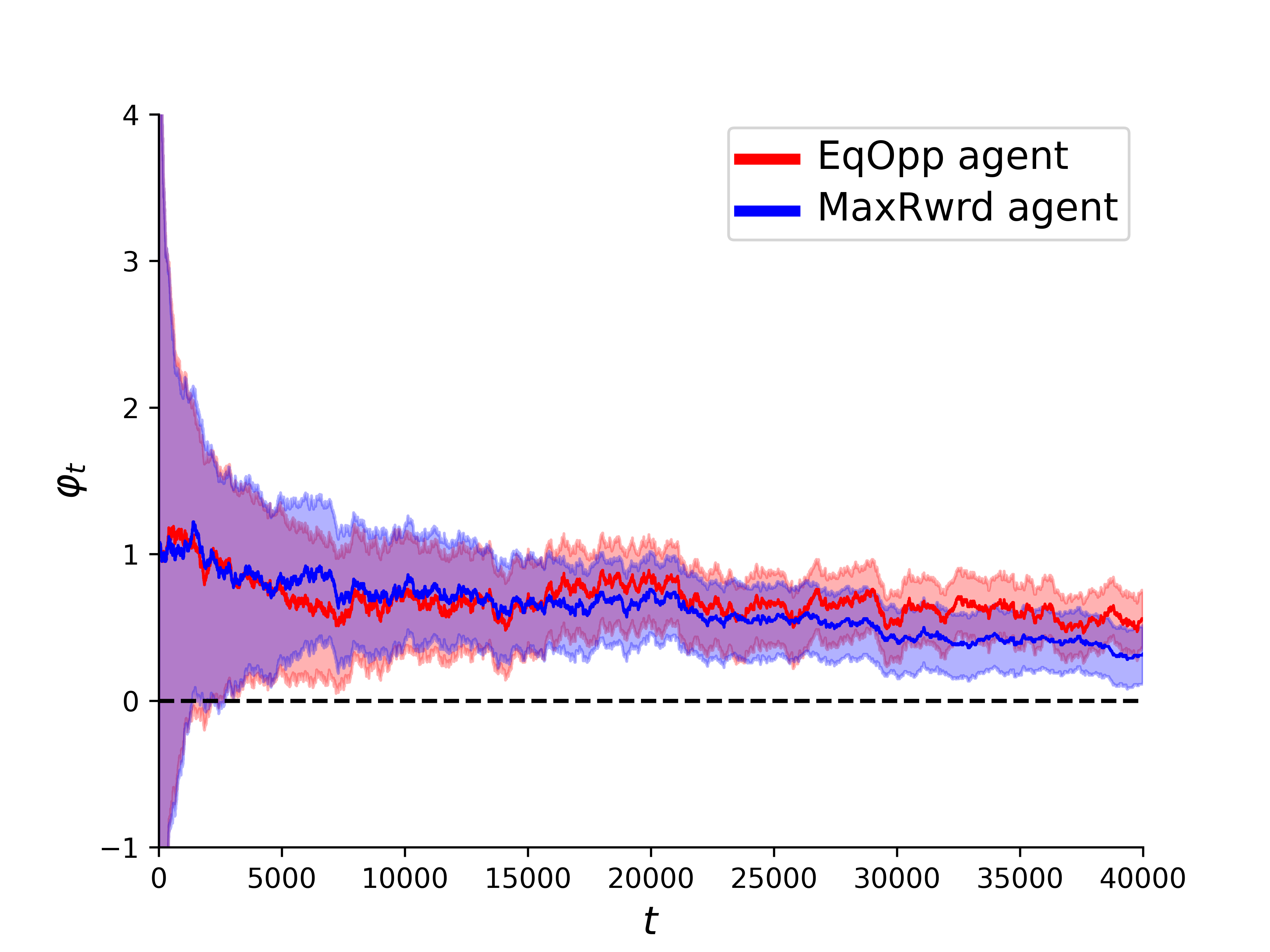}
	\includegraphics[scale=0.3]{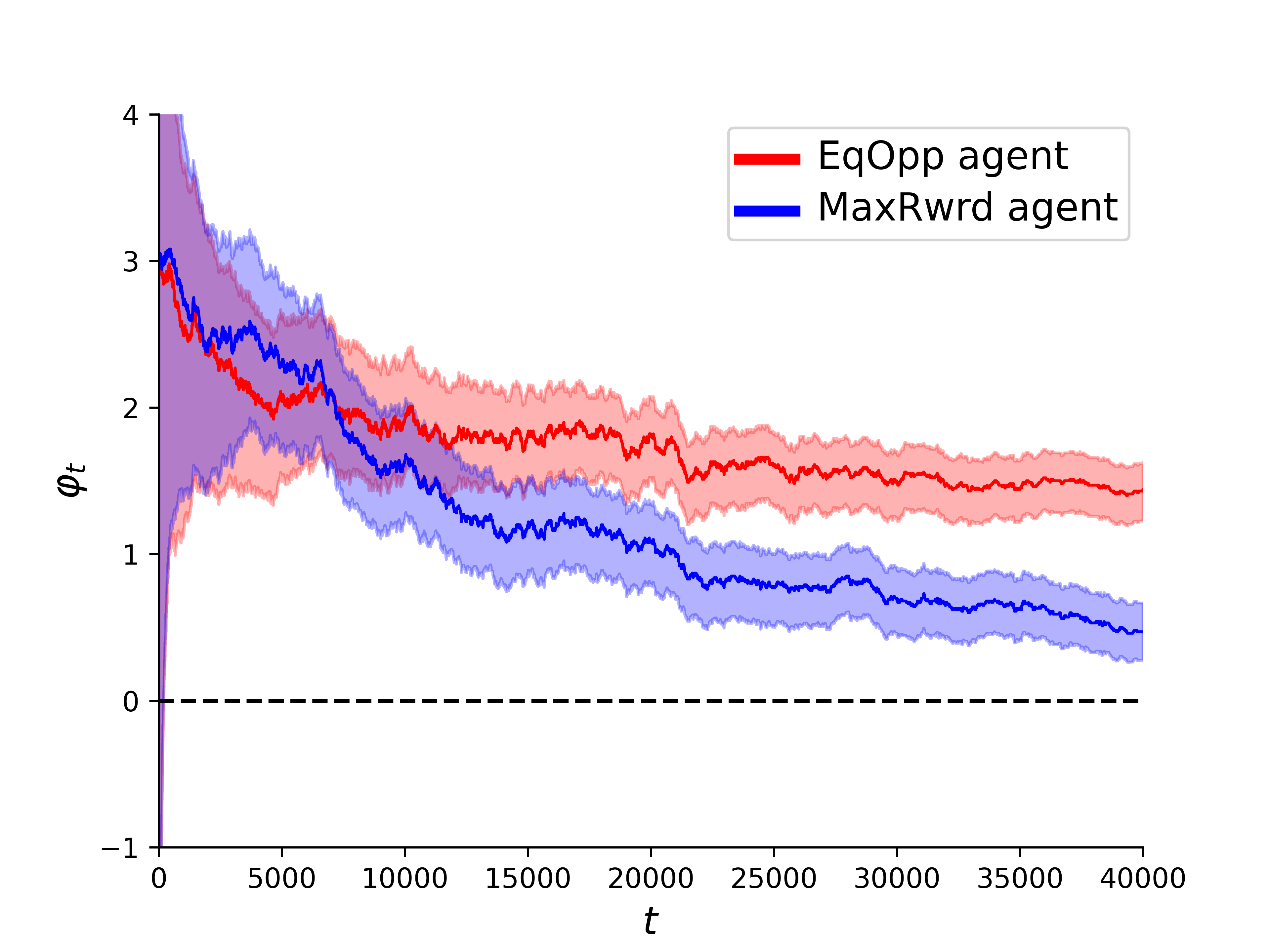}
	\caption{Output estimates of the monitors at each time step on simulated trajectories for the lending example obtained from \texttt{ml-fairness-gym}.
	The three plots are ordered from left to right in the increasing order of initial bias. 
	For each case, we considered two different policies of the agent: the \texttt{MaxRwrd} agent (blue) maximizes its own reward without trying to optimize any fairness criterion, whereas the \texttt{EqOpp} agent (red) also tries to ensure equalized opportunity statically (i.e., in its one-shot decisions).
	The shaded regions are the intervals computed by the monitor \textsf{LendingMonitor}, whereas the solid lines are the (unknown to the monitor) true values of the properties. 
	The horizontal dashed line corresponds to the perfectly fair scenario (i.e., $\varphi_t = 0$).
	}
	\label{fig:lending monitor:plots}
\end{figure*}

We summarize the outputs of our monitor in Fig.~\ref{fig:lending monitor:plots}.
We observe that our monitor's outputs match with the correct value of bias in the system at every point in time.
Moreover, the monitor \textsf{LendingMonitor} took 
on an average \num{3}\,\si{\micro\second} to compute a new estimate after receiving every new observation.

Interestingly, we note that the \texttt{MaxRwrd} agent's policy becomes more fair in the long run compared to the \texttt{EqOpp} agent's policy.
This phenomenon has already been demonstrated in the existing literature \cite{d2020fairness}.
With our monitors, it is actually possible to detect such biases in real-time, without any information or assumption about what policy the bank is using or how the individuals will react (i.e., whether they repay or not).


\section{A Dynamic Attention Allocation Problem}
\label{sec:attention}

Now we consider the attention allocation problem.
Suppose there are $L\geq 2$ locations, and at each location and at each time step, a number of incidents take place.
There is a machine-learned allocator who needs to allocate its limited amount of resources to the locations in order to discover the incidents, where every event needs one unit of attention to be discovered.
We design a runtime monitor to check, in real-time, if the allocator is fair in allocating its limited amount of attention among two particular locations $A$ and $B$.
Suppose in each location and at every time step, some number of events appear according to the Poisson distribution with unknown parameters. 
At any location and at any given time, the rate with which events appear is inversely proportional to the number of attention units allocated to that location at the last time step; the exact relationship will be formalized in Sec.~\ref{sec:attention:problem formulation}.
We assume that the allocator has knowledge about this relationship.
The fairness criteria that we wish to monitor requires that the probability with which any event will be discovered across the two locations should be close to each other.

We streamline our exposition on a simpler instance of the original problem that was studied by D’Amour et al.\ \cite{d2020fairness}, where they considered the fairness measure as the maximum pairwise disparity in discovery probabilities.
We point out that this general property can also be handled using our monitors, by simply having a separate monitor for each pair of locations, and then aggregating the outputs of all the monitors using interval arithmetic and union bound. 

\subsection{Problem Formulation}
\label{sec:attention:problem formulation}

\noindent\textbf{The agent and the environment:}
Here, the two locations $A$ and $B$ are the two groups, as well as the only members in the respective groups.
The feature of each location is the number of events (in $\mathbb{N}$).
At every time step $t$, the environment samples a pair of (random) features $X_t^A$ and $X_t^B$ for the two locations, such that $X_t^g \sim \poisson{\lambda_t^g}$, where $\poisson{\lambda_t^g}$ represents the Poisson distribution with parameter $\lambda_t^g$.
The numbers of incidents in the two locations are given by $X_t^A+1$ and $X_t^B+1$, ensuring that the minimum number of events is $1$, which is necessary for technical reasons (will be explained later).
Observe that, we slightly deviate from the setting that we introduced in Sec.~\ref{sec:general problem}, in that we obtain features of two individuals (i.e., the locations) from both groups simultaneously at each time step.
Technique-wise, this is not a problem since we are going to use two separate \ExpMon monitors for the two locations. 
Notation-wise, this is simpler, as otherwise we would need \emph{vectors}-valued features.

The agent's action is a (random) vector $(Y_t^A,Y_t^B) \in \mathbb{N}^2$ to allocate its $K$ units of attention to the respective locations, i.e. $ Y^A_t + Y^B_t \leq K$.
The entry $Y^{g}_t$ represents the number of attention units allocated to location $g$ at time $t$. 
In this example, the reaction of the environment to the agent's action has no role, i.e., we set $Z_t=\bot$ (a dummy symbol).

As usual, we assume a change function $\Delta_t$ is given (defined below) that causes a shift in the expected value of $X_t$.
Since the expected value of $\poisson{\lambda_t^g}$ is $\lambda_t^g$, this corresponds to a change in the Poisson parameter, causing changes in the rate at which events appear in the respective locations in the next step.
Given a fixed parameter $\gamma>0$, which controls how dynamic the system is, the change function is defined as:
\begin{align*}
	\forall g\in \set{A,B}\;.\;\Delta(g,x_t^g,y_t^g) = \begin{cases}
															\gamma	& \text{if } y_t^g = 0,\\
															-\gamma\cdot y_t^g & \text{otherwise},\\
														\end{cases}
\end{align*}
where we drop the fourth argument  $z_t$ from $\Delta_t$ as it has no role.

\smallskip
\noindent\textbf{The fairness property:}
The well-being function $f$ in \eqref{equ:expected well-being} in this example is called the discovery probability, and we want to monitor its disparity between the two locations.
The discovery probability at time $t$ and in location $g$ can be formalized as the expected value of the ratio of the discovered events $\overline{X}^{g}_t = \min\set{X^{g}_t+1,Y^{g}_t}$ to the  the actual number of events $X^{g}_t+1$.
Notice that had we defined the number of events as $X^{g}_t$, discovery probability would be undefined (since $X^g_t$ can be zero).
The discovery probability at time $t$ for a given observation sequence $\vec{o}_t$ can also be represented as the following conditional expectation:
\begin{align*}
	\omega_t^g\coloneqq \expe_g\left(\frac{\overline{X}^g_{t}}{X^g_{t}+1}\; \middle| \vec{o}_{t-1}, y_t^g \right).
\end{align*}
The time-varying social fairness criteria at every time $t$ is given by $\varphi_t\coloneqq \omega_t^A  - \omega_t^B$.

\subsection{The Runtime Monitor}
We show that $\omega_t^g$ has a closed-form expression $\eta(y_t^g,\lambda_t^g)$ for a given concrete $\lambda_t^g$ and $y_t^g$, 
where
\begin{align} \label{eq:eta}
	\eta(y, \lambda)\coloneqq e^{-\lambda} \sum_{k=0}^{y-1}\frac{\lambda^k}{k!} \left( 1- \frac{y}{k+1}\right) + \frac{y}{\lambda}\left(1-e^{-\lambda} \right).
\end{align}
Furthermore, we show that, for a fixed $y$, the function $\eta(y,\cdot)$ is strictly decreasing everywhere in the positive reals 
\ifarxiv
	(see App.~\ref{theorem:eta_decreasing}).
\fi
This property of $\eta(y,\cdot)$ enables us to efficiently compute an interval estimate of $\eta(y,\cdot)$ and, in turn $\omega_t^g$, from an interval estimate of $\lambda$.

Using these auxiliary results, we construct the monitor as follows.
We first use the general monitor \ExpMon from Alg.~\ref{alg:general monitor:single random variable} to estimate, for each location $g$, the expected value $\expe_g(X_t^g)$, which is the same as $\lambda_t^g$ (follows from the property of the Poisson distributions).
We make two mild assumptions:
First, we assume that the Poisson parameters in both locations are bounded between two positive reals $\underline{\lambda}<\overline{\lambda}$, allowing us to establish that $X_t^g$ is a sub-exponential random variable with parameters $(2\lambda,2)$ 
\ifarxiv
	(see App.~\ref{lemma:subexponential_poisson}).
\fi
Second, we assume that the sequence of observations are such that the parameter would not reach zero no matter what its true initial value was. 
Otherwise the function $\Delta$ would no longer reflect the differences in the expected values. 
This can be checked by the monitor at runtime, by checking if the parameter would reach zero had it started from $\underline{\lambda}$ in the worst case (the value closest to zero). 
We omit the check for simplicity.

Suppose $[\lambda_{\mathsf{min}}^g,\lambda_{\mathsf{max}}^g]$ is the interval output by \ExpMon as the estimate of $\lambda_t^g$.
Then from $\omega_t^g= \eta(y_t^g,\lambda_t^g)$ and the strictly decreasing property of $\eta$, we obtain the corresponding interval estimate for $\omega_t^g$ as $[\eta(y_t^g,\lambda_{\mathsf{max}}),\eta(y_t^g,\lambda_{\mathsf{min}})]$.
From the interval estimates for the discovery probability of each group, we obtain the overall fairness estimate by computing the interval difference, as we did for the lending monitor.
The pseudocode of the monitor is in Alg.~\ref{alg:attention monitor}.

\begin{algorithm}
	\caption{\textsf{AttentionMonitor}}
	\label{alg:attention monitor}
	\begin{minipage}{0.52\textwidth}
		\begin{algorithmic}[1]
			\Function{$\mathtt{Init}$}{$\Delta,\delta,\overline{\lambda}$}
 				\State $M_A \gets \ExpMon.\mathtt{Init}(\Delta,\frac{\delta}{2},2\overline{\lambda},2)$
				\State $M_B \gets \ExpMon.\mathtt{Init}(\Delta,\frac{\delta}{2},2\overline{\lambda},2)$
			\EndFunction
		\end{algorithmic}
	\end{minipage}
	\begin{minipage}{0.47\textwidth}
		\begin{algorithmic}[1]
			\Function{$\mathtt{Compute}$}{$(x^A,x^B),(y^A,y^B)$}
				\State $[\lambda^A_{\mathsf{min}},\lambda^A_{\mathsf{max}}] \gets M_A.\mathtt{Compute}(x^A,A,y^A,\bot)$
				\State $[\lambda^B_{\mathsf{min}},\lambda^B_{\mathsf{max}}] \gets M_B.\mathtt{Compute}(x^B,B,y^B,\bot)$
				\State $\overline{E}_A\gets \left[\eta( y^A, \lambda^A_{\mathsf{max}}),\eta(y^A,\lambda^A_{\mathsf{min}})\right] $
				\State $\overline{E}_B\gets \left[\eta( y^B, \lambda^B_{\mathsf{max}}),\eta(y^B,\lambda^B_{\mathsf{min}})\right] $
				\State $\overline{\varphi} \gets \overline{E}_A - \overline{E}_B$ \Comment{interval difference}
				\State \Return $\overline{\varphi}$
			\EndFunction
		\end{algorithmic}
	\end{minipage}
\end{algorithm}

\subsection{Experimental Outcome}

We summarize the outputs of our monitor in Fig.~\ref{fig:attention monitor:plots}.
We consider three types of agent, uniform, greedy, and fair-constrained greed (with $\alpha=0.75)$.
While the uniform agent behaves randomly, without taking into account the actual incidence rate, the greedy agent tries to minimize the chances of missed discoveries by keeping an estimate of the incidence rates.
The constrained greedy agent needs to additionally ensure a fairness criteria.
Our experiments empirically show that no matter what the experimental conditions are, our monitor is able to provide real-time information about the time-varying biases in the system.
Moreover, the monitor \textsf{AttentionMonitor} took 
on an average \num{28}\,\si{\micro\second} to compute a new estimate after receiving every new observation.
The experiments demonstrate the practical usefulness of our monitors.

\begin{figure*}
	\includegraphics[scale=0.3]{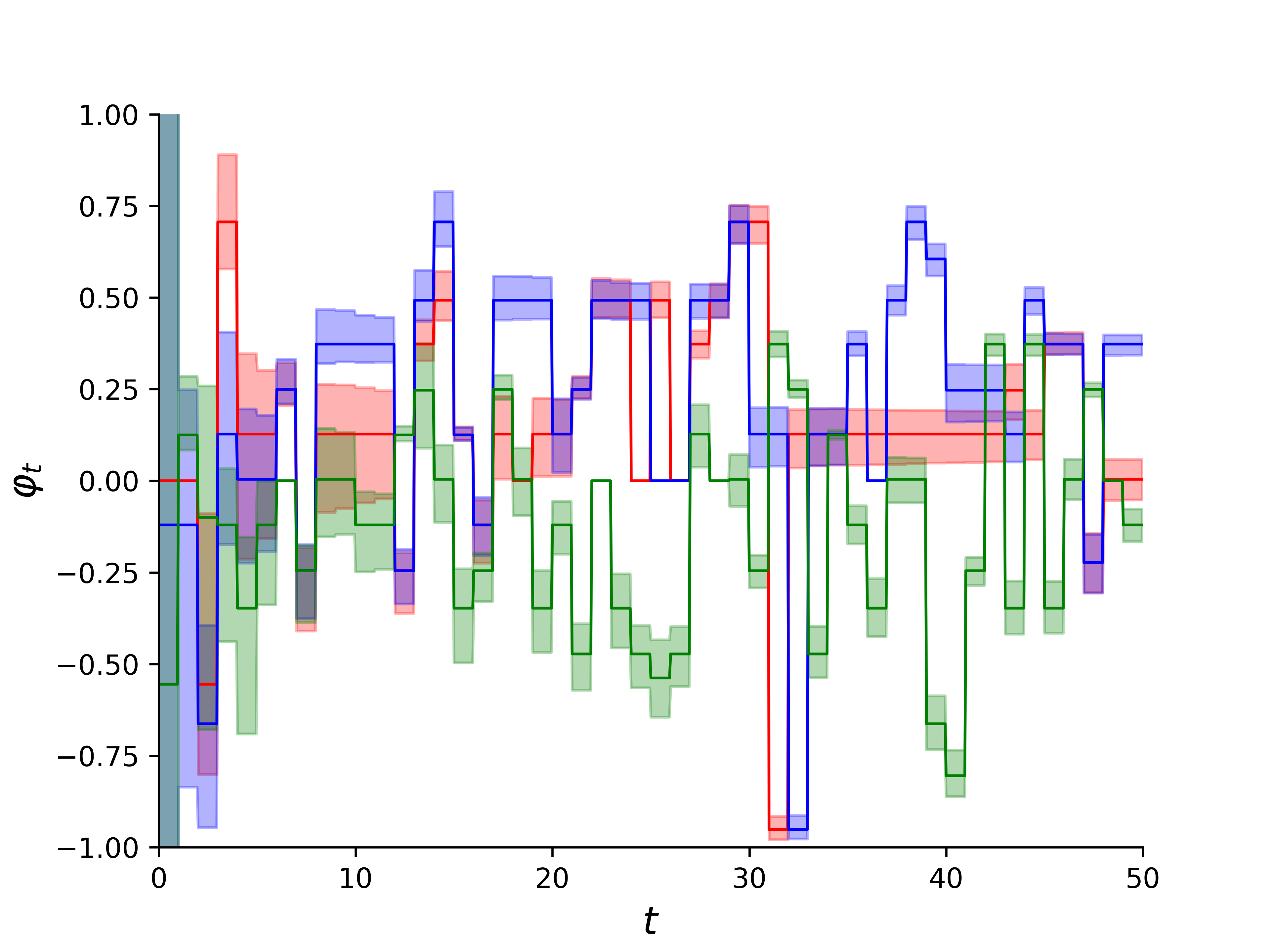}
	\includegraphics[scale=0.3]{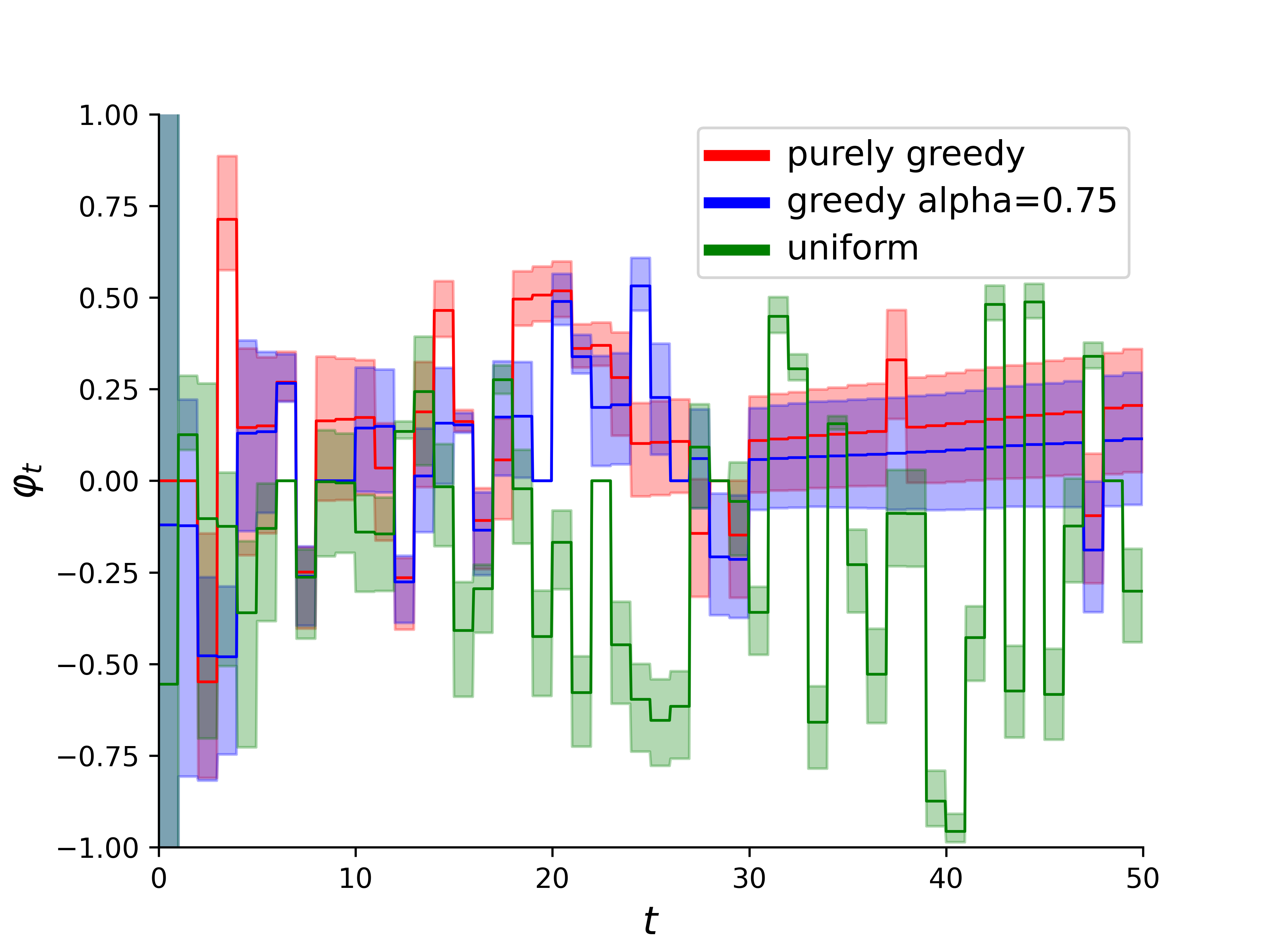}
		\includegraphics[scale=0.3]{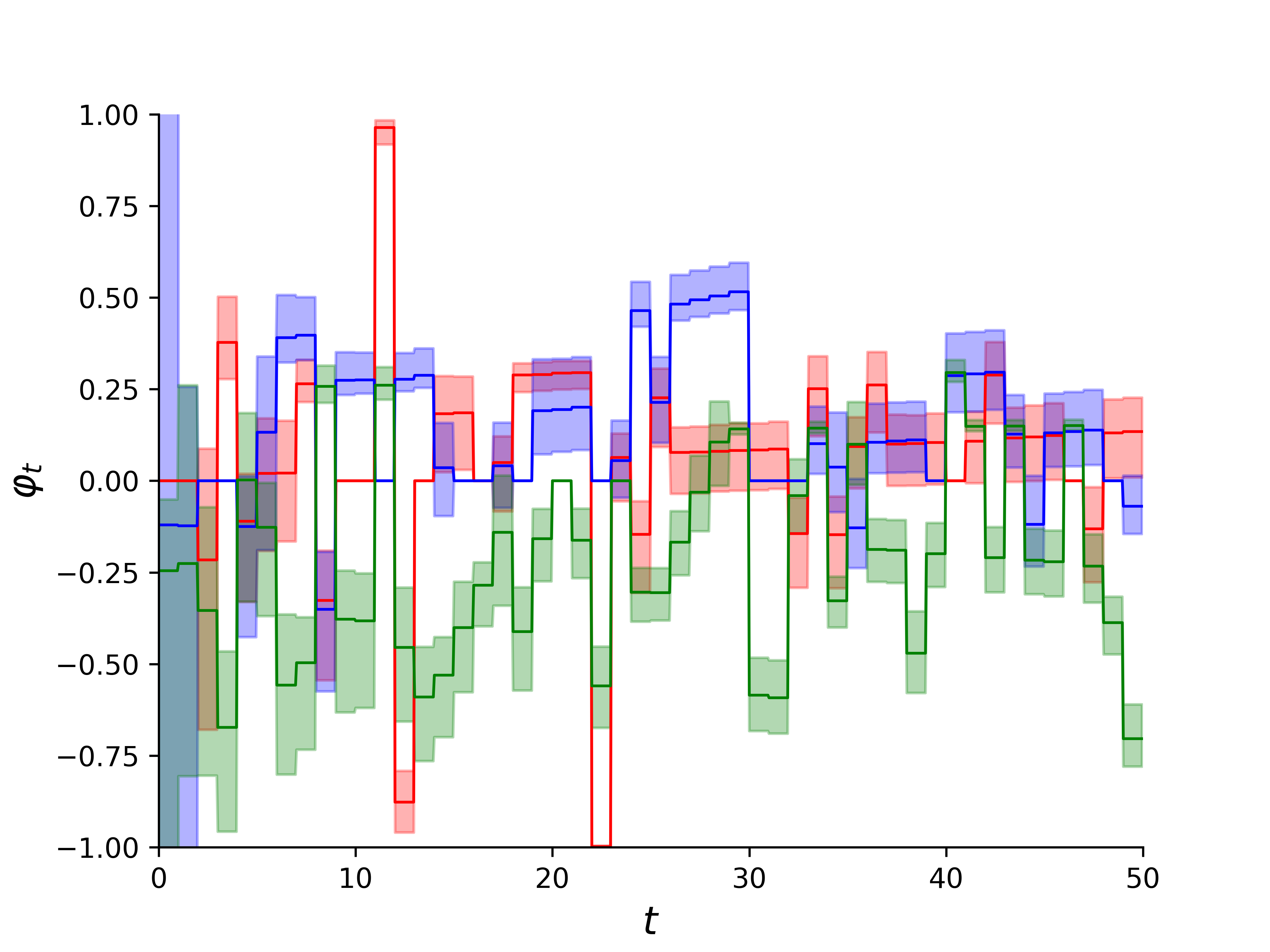}
	\caption{Output estimates of the monitors at each time step on simulation traces for the attention allocation example obtained from \texttt{ml-fairness-gym}. 
	Left: $L=5$ (no.\ of locations), $K=6$ (total units of attention), $\gamma = 0$ (change in the Poisson parameters).
	Middle: $L=5$, $K=6$, $\gamma = 0.0025$.
	Right: $L=10$, $K=10$, $\gamma = 0.0025$.
	For each case, we considered three different policies of the agent (description in the text). 
	The shaded regions are the intervals computed by the monitor \textsf{AttentionMonitor}, whereas the solid lines are the (unknown) true values of the properties.
	}
	\label{fig:attention monitor:plots}
\end{figure*}
\section{Conclusion and Future Work}
\label{sec:conclusion}

We present an approach for real-time monitoring of the long-run fairness of machine-learned agents deployed in dynamic environments.
Our monitors observe a long sequence of events generated from the interactions between the agent and its environment, and output, after each new observation, a quantitative PAC-style estimate of how fair or biased the agent's policy was until that point in time. 
The strength of our monitors is their ability to compute interval estimates of fairness values in the face of frequent changes in the underlying distribution, a setting that prevents a static estimation of fairness at each time step.
The presented method allows for the computation of confidence intervals when the monitored random variable is sub-exponential. By extension, they can handle sub-Gaussian random variables as well.
Using a prototype implementation, we demonstrated the practical usefulness of the monitoring approach on examples taken from the literature.

We took great effort to ensure that the interval estimates computed by the monitor 
hold non-asymptotically.
Consequently, we avoided a direct comparison with methods that rely on the central limit theorem. 
However, we acknowledge that loosening this restriction will allow for a wider range of applications.
Furthermore, we showed computations of PAC estimates on fairness properties with some specific well-being functions, such as the expected credit score in the lending problem and the discovery probability in the dynamic attention allocation problem, etc.
Although an extension to general well-being functions is difficult, a generalization to restricted classes of well-being functions is conceivable: for instance, when the well-being function is in the form of arithmetic expressions, we can use interval arithmetic to deduce the overall PAC bound, or when the well-being function is convex, we can use convex optimization to deduce tight PAC estimates, etc. 

We see several immediate future directions.
Firstly, we considered only one particular class of fairness properties, namely the ones which can be written as a difference between expected values of a given function of population parameters across two groups.
Investigating other classes of properties will be an important goal.
We will be able to extend our monitors to handle individuals with multiple features, i.e.\ feature vectors, without adding any additional technical machinery.
Secondly, we assumed perfect observability about the actions of the agent and the environment, whereas in reality actions are often either partially observable or the observations are noisy (e.g., the college admission example in D’Amour et al.\ \cite{d2020fairness}).
Thus extensions to partially observed and noisy models will be helpful.
Thirdly, our monitors require, at least partially, information about the changes in the system. A natural extension, would be to relax this condition.
Finally, we only considered the monitoring, i.e., the problem of checking fairness in real-time.
The next step will be to combine monitoring and intervention, so that we obtain an automated procedure for controlling dynamic fairness properties of a machine-learned agent.

	\begin{acks}
	The authors would like to thank the anonymous reviewers for
	their valuable comments and helpful suggestions. This work is
	supported by the \grantsponsor{ERC-2020-AdG}{European Research Council}{} under Grant
	No.:~\grantnum{ERC-2020-AdG}{ERC-2020-AdG 101020093}.
	\end{acks}
	
	\bibliographystyle{ACM-Reference-Format}
	\bibliography{references}


\begin{thebibliography}{39}


\ifx \showCODEN    \undefined \def \showCODEN     #1{\unskip}     \fi
\ifx \showDOI      \undefined \def \showDOI       #1{#1}\fi
\ifx \showISBNx    \undefined \def \showISBNx     #1{\unskip}     \fi
\ifx \showISBNxiii \undefined \def \showISBNxiii  #1{\unskip}     \fi
\ifx \showISSN     \undefined \def \showISSN      #1{\unskip}     \fi
\ifx \showLCCN     \undefined \def \showLCCN      #1{\unskip}     \fi
\ifx \shownote     \undefined \def \shownote      #1{#1}          \fi
\ifx \showarticletitle \undefined \def \showarticletitle #1{#1}   \fi
\ifx \showURL      \undefined \def \showURL       {\relax}        \fi
\providecommand\bibfield[2]{#2}
\providecommand\bibinfo[2]{#2}
\providecommand\natexlab[1]{#1}
\providecommand\showeprint[2][]{arXiv:#2}

\bibitem[Agarwal et~al\mbox{.}(2018)]%
        {agarwal2018reductions}
\bibfield{author}{\bibinfo{person}{Alekh Agarwal}, \bibinfo{person}{Alina
  Beygelzimer}, \bibinfo{person}{Miroslav Dud{\'\i}k}, \bibinfo{person}{John
  Langford}, {and} \bibinfo{person}{Hanna Wallach}.}
  \bibinfo{year}{2018}\natexlab{}.
\newblock \showarticletitle{A reductions approach to fair classification}. In
  \bibinfo{booktitle}{\emph{International Conference on Machine Learning}}.
  PMLR, \bibinfo{pages}{60--69}.
\newblock


\bibitem[Albarghouthi and Vinitsky(2019)]%
        {albarghouthi2019fairness}
\bibfield{author}{\bibinfo{person}{Aws Albarghouthi} {and}
  \bibinfo{person}{Samuel Vinitsky}.} \bibinfo{year}{2019}\natexlab{}.
\newblock \showarticletitle{Fairness-aware programming}. In
  \bibinfo{booktitle}{\emph{Proceedings of the Conference on Fairness,
  Accountability, and Transparency}}. \bibinfo{pages}{211--219}.
\newblock


\bibitem[Bartocci and Falcone(2018)]%
        {bartocci2018lectures}
\bibfield{author}{\bibinfo{person}{Ezio Bartocci} {and}
  \bibinfo{person}{Yli{\`e}s Falcone}.} \bibinfo{year}{2018}\natexlab{}.
\newblock \bibinfo{booktitle}{\emph{Lectures on Runtime Verification}}.
\newblock \bibinfo{publisher}{Springer}.
\newblock


\bibitem[Berk et~al\mbox{.}(2017)]%
        {berk2017convex}
\bibfield{author}{\bibinfo{person}{Richard Berk}, \bibinfo{person}{Hoda
  Heidari}, \bibinfo{person}{Shahin Jabbari}, \bibinfo{person}{Matthew Joseph},
  \bibinfo{person}{Michael Kearns}, \bibinfo{person}{Jamie Morgenstern},
  \bibinfo{person}{Seth Neel}, {and} \bibinfo{person}{Aaron Roth}.}
  \bibinfo{year}{2017}\natexlab{}.
\newblock \showarticletitle{A convex framework for fair regression}.
\newblock \bibinfo{journal}{\emph{arXiv preprint arXiv:1706.02409}}
  (\bibinfo{year}{2017}).
\newblock


\bibitem[Bloem et~al\mbox{.}(2015)]%
        {bloem2015shield}
\bibfield{author}{\bibinfo{person}{Roderick Bloem}, \bibinfo{person}{Bettina
  K{\"o}nighofer}, \bibinfo{person}{Robert K{\"o}nighofer}, {and}
  \bibinfo{person}{Chao Wang}.} \bibinfo{year}{2015}\natexlab{}.
\newblock \showarticletitle{Shield synthesis: Runtime enforcement for reactive
  systems}. In \bibinfo{booktitle}{\emph{Tools and Algorithms for the
  Construction and Analysis of Systems: 21st International Conference, TACAS
  2015, Held as Part of the European Joint Conferences on Theory and Practice
  of Software, ETAPS 2015, London, UK, April 11-18, 2015, Proceedings 21}}.
  Springer, \bibinfo{pages}{533--548}.
\newblock


\bibitem[Calders and {\v{Z}}liobait{\.e}(2013)]%
        {calders2013unbiased}
\bibfield{author}{\bibinfo{person}{Toon Calders} {and}
  \bibinfo{person}{Indr{\.e} {\v{Z}}liobait{\.e}}.}
  \bibinfo{year}{2013}\natexlab{}.
\newblock \showarticletitle{Why unbiased computational processes can lead to
  discriminative decision procedures}.
\newblock \bibinfo{journal}{\emph{Discrimination and Privacy in the Information
  Society: Data mining and profiling in large databases}}
  (\bibinfo{year}{2013}), \bibinfo{pages}{43--57}.
\newblock


\bibitem[Chaney et~al\mbox{.}(2018)]%
        {chaney2018algorithmic}
\bibfield{author}{\bibinfo{person}{Allison~JB Chaney},
  \bibinfo{person}{Brandon~M Stewart}, {and} \bibinfo{person}{Barbara~E
  Engelhardt}.} \bibinfo{year}{2018}\natexlab{}.
\newblock \showarticletitle{How algorithmic confounding in recommendation
  systems increases homogeneity and decreases utility}. In
  \bibinfo{booktitle}{\emph{Proceedings of the 12th ACM conference on
  recommender systems}}. \bibinfo{pages}{224--232}.
\newblock


\bibitem[Chen et~al\mbox{.}(2020)]%
        {chen2020fair}
\bibfield{author}{\bibinfo{person}{Yifang Chen}, \bibinfo{person}{Alex
  Cuellar}, \bibinfo{person}{Haipeng Luo}, \bibinfo{person}{Jignesh Modi},
  \bibinfo{person}{Heramb Nemlekar}, {and} \bibinfo{person}{Stefanos
  Nikolaidis}.} \bibinfo{year}{2020}\natexlab{}.
\newblock \showarticletitle{The fair contextual multi-armed bandit}. In
  \bibinfo{booktitle}{\emph{Proceedings of the 19th International Conference on
  Autonomous Agents and MultiAgent Systems}}.
\newblock


\bibitem[D'Amour et~al\mbox{.}(2020)]%
        {d2020fairness}
\bibfield{author}{\bibinfo{person}{Alexander D'Amour}, \bibinfo{person}{Hansa
  Srinivasan}, \bibinfo{person}{James Atwood}, \bibinfo{person}{Pallavi
  Baljekar}, \bibinfo{person}{David Sculley}, {and} \bibinfo{person}{Yoni
  Halpern}.} \bibinfo{year}{2020}\natexlab{}.
\newblock \showarticletitle{Fairness is not static: deeper understanding of
  long term fairness via simulation studies}. In
  \bibinfo{booktitle}{\emph{Proceedings of the 2020 Conference on Fairness,
  Accountability, and Transparency}}. \bibinfo{pages}{525--534}.
\newblock


\bibitem[Dwork et~al\mbox{.}(2012)]%
        {dwork2012fairness}
\bibfield{author}{\bibinfo{person}{Cynthia Dwork}, \bibinfo{person}{Moritz
  Hardt}, \bibinfo{person}{Toniann Pitassi}, \bibinfo{person}{Omer Reingold},
  {and} \bibinfo{person}{Richard Zemel}.} \bibinfo{year}{2012}\natexlab{}.
\newblock \showarticletitle{Fairness through awareness}. In
  \bibinfo{booktitle}{\emph{Proceedings of the 3rd innovations in theoretical
  computer science conference}}. \bibinfo{pages}{214--226}.
\newblock


\bibitem[Elzayn et~al\mbox{.}(2019)]%
        {elzayn2019fair}
\bibfield{author}{\bibinfo{person}{Hadi Elzayn}, \bibinfo{person}{Shahin
  Jabbari}, \bibinfo{person}{Christopher Jung}, \bibinfo{person}{Michael
  Kearns}, \bibinfo{person}{Seth Neel}, \bibinfo{person}{Aaron Roth}, {and}
  \bibinfo{person}{Zachary Schutzman}.} \bibinfo{year}{2019}\natexlab{}.
\newblock \showarticletitle{Fair algorithms for learning in allocation
  problems}. In \bibinfo{booktitle}{\emph{Proceedings of the Conference on
  Fairness, Accountability, and Transparency}}. \bibinfo{pages}{170--179}.
\newblock


\bibitem[Ensign et~al\mbox{.}(2018)]%
        {ensign2018runaway}
\bibfield{author}{\bibinfo{person}{Danielle Ensign}, \bibinfo{person}{Sorelle~A
  Friedler}, \bibinfo{person}{Scott Neville}, \bibinfo{person}{Carlos
  Scheidegger}, {and} \bibinfo{person}{Suresh Venkatasubramanian}.}
  \bibinfo{year}{2018}\natexlab{}.
\newblock \showarticletitle{Runaway feedback loops in predictive policing}. In
  \bibinfo{booktitle}{\emph{Conference on Fairness, Accountability and
  Transparency}}. PMLR, \bibinfo{pages}{160--171}.
\newblock


\bibitem[Feldman et~al\mbox{.}(2015)]%
        {feldman2015certifying}
\bibfield{author}{\bibinfo{person}{Michael Feldman}, \bibinfo{person}{Sorelle~A
  Friedler}, \bibinfo{person}{John Moeller}, \bibinfo{person}{Carlos
  Scheidegger}, {and} \bibinfo{person}{Suresh Venkatasubramanian}.}
  \bibinfo{year}{2015}\natexlab{}.
\newblock \showarticletitle{Certifying and removing disparate impact}. In
  \bibinfo{booktitle}{\emph{proceedings of the 21th ACM SIGKDD international
  conference on knowledge discovery and data mining}}.
  \bibinfo{pages}{259--268}.
\newblock


\bibitem[Ferr{\`e}re et~al\mbox{.}(2020)]%
        {ferrere2019monitoring}
\bibfield{author}{\bibinfo{person}{Thomas Ferr{\`e}re},
  \bibinfo{person}{Thomas~A. Henzinger}, {and} \bibinfo{person}{Bernhard
  Kragl}.} \bibinfo{year}{2020}\natexlab{}.
\newblock \showarticletitle{{Monitoring Event Frequencies}}. In
  \bibinfo{booktitle}{\emph{28th EACSL Annual Conference on Computer Science
  Logic (CSL 2020)}} \emph{(\bibinfo{series}{Leibniz International Proceedings
  in Informatics (LIPIcs)}, Vol.~\bibinfo{volume}{152})},
  \bibfield{editor}{\bibinfo{person}{Maribel Fern{\'a}ndez} {and}
  \bibinfo{person}{Anca Muscholl}} (Eds.). \bibinfo{publisher}{Schloss
  Dagstuhl--Leibniz-Zentrum fuer Informatik}, \bibinfo{address}{Dagstuhl,
  Germany}, \bibinfo{pages}{20:1--20:16}.
\newblock
\showISBNx{978-3-95977-132-0}
\showISSN{1868-8969}
\urldef\tempurl%
\url{https://doi.org/10.4230/LIPIcs.CSL.2020.20}
\showDOI{\tempurl}


\bibitem[Fuster et~al\mbox{.}(2022)]%
        {fuster2022predictably}
\bibfield{author}{\bibinfo{person}{Andreas Fuster}, \bibinfo{person}{Paul
  Goldsmith-Pinkham}, \bibinfo{person}{Tarun Ramadorai}, {and}
  \bibinfo{person}{Ansgar Walther}.} \bibinfo{year}{2022}\natexlab{}.
\newblock \showarticletitle{Predictably unequal? The effects of machine
  learning on credit markets}.
\newblock \bibinfo{journal}{\emph{The Journal of Finance}}
  \bibinfo{volume}{77}, \bibinfo{number}{1} (\bibinfo{year}{2022}),
  \bibinfo{pages}{5--47}.
\newblock


\bibitem[Gordaliza et~al\mbox{.}(2019)]%
        {gordaliza2019obtaining}
\bibfield{author}{\bibinfo{person}{Paula Gordaliza}, \bibinfo{person}{Eustasio
  Del~Barrio}, \bibinfo{person}{Gamboa Fabrice}, {and}
  \bibinfo{person}{Jean-Michel Loubes}.} \bibinfo{year}{2019}\natexlab{}.
\newblock \showarticletitle{Obtaining fairness using optimal transport theory}.
  In \bibinfo{booktitle}{\emph{International Conference on Machine Learning}}.
  PMLR, \bibinfo{pages}{2357--2365}.
\newblock


\bibitem[Hardt et~al\mbox{.}(2016)]%
        {hardt2016equality}
\bibfield{author}{\bibinfo{person}{Moritz Hardt}, \bibinfo{person}{Eric Price},
  {and} \bibinfo{person}{Nati Srebro}.} \bibinfo{year}{2016}\natexlab{}.
\newblock \showarticletitle{Equality of opportunity in supervised learning}.
\newblock \bibinfo{journal}{\emph{Advances in neural information processing
  systems}}  \bibinfo{volume}{29} (\bibinfo{year}{2016}).
\newblock


\bibitem[Harwell(2018)]%
        {Harwell2018Alexa}
\bibfield{author}{\bibinfo{person}{Drew Harwell}.}
  \bibinfo{year}{2018}\natexlab{}.
\newblock \bibinfo{title}{Amazon’s Alexa and Google Home show accent bias,
  with Chinese and Spanish hardest to understand}.
\newblock \bibinfo{howpublished}{\url{http://bit.ly/2QFA1MR}}.
\newblock
\newblock
\shownote{Accessed: 05.02.2023}.


\bibitem[Hashimoto et~al\mbox{.}(2018)]%
        {hashimoto2018fairness}
\bibfield{author}{\bibinfo{person}{Tatsunori Hashimoto}, \bibinfo{person}{Megha
  Srivastava}, \bibinfo{person}{Hongseok Namkoong}, {and}
  \bibinfo{person}{Percy Liang}.} \bibinfo{year}{2018}\natexlab{}.
\newblock \showarticletitle{Fairness without demographics in repeated loss
  minimization}. In \bibinfo{booktitle}{\emph{International Conference on
  Machine Learning}}. PMLR, \bibinfo{pages}{1929--1938}.
\newblock


\bibitem[Heidari et~al\mbox{.}(2019)]%
        {heidari2019long}
\bibfield{author}{\bibinfo{person}{Hoda Heidari}, \bibinfo{person}{Vedant
  Nanda}, {and} \bibinfo{person}{Krishna~P Gummadi}.}
  \bibinfo{year}{2019}\natexlab{}.
\newblock \showarticletitle{On the long-term impact of algorithmic decision
  policies: Effort unfairness and feature segregation through social learning}.
\newblock \bibinfo{journal}{\emph{arXiv preprint arXiv:1903.01209}}
  (\bibinfo{year}{2019}).
\newblock


\bibitem[Henzinger et~al\mbox{.}(2023)]%
        {henzinger2023monitoring}
\bibfield{author}{\bibinfo{person}{Thomas~A.\ Henzinger},
  \bibinfo{person}{Mahyar Karimi}, \bibinfo{person}{Konstantin Kueffner}, {and}
  \bibinfo{person}{Kaushik Mallik}.} \bibinfo{year}{2023}\natexlab{}.
\newblock \showarticletitle{Monitoring Algorithmic Fairness}. In
  \bibinfo{booktitle}{\emph{Computer Aided Verification}}.
\newblock
\newblock
\shownote{(to appear)}.


\bibitem[Hossain et~al\mbox{.}(2021)]%
        {hossain2021fair}
\bibfield{author}{\bibinfo{person}{Safwan Hossain}, \bibinfo{person}{Evi
  Micha}, {and} \bibinfo{person}{Nisarg Shah}.}
  \bibinfo{year}{2021}\natexlab{}.
\newblock \showarticletitle{Fair algorithms for multi-agent multi-armed
  bandits}.
\newblock \bibinfo{journal}{\emph{Advances in Neural Information Processing
  Systems}}  \bibinfo{volume}{34} (\bibinfo{year}{2021}),
  \bibinfo{pages}{24005--24017}.
\newblock


\bibitem[Kamiran and Calders(2012)]%
        {kamiran2012data}
\bibfield{author}{\bibinfo{person}{Faisal Kamiran} {and} \bibinfo{person}{Toon
  Calders}.} \bibinfo{year}{2012}\natexlab{}.
\newblock \showarticletitle{Data preprocessing techniques for classification
  without discrimination}.
\newblock \bibinfo{journal}{\emph{Knowledge and information systems}}
  \bibinfo{volume}{33}, \bibinfo{number}{1} (\bibinfo{year}{2012}),
  \bibinfo{pages}{1--33}.
\newblock


\bibitem[Liu et~al\mbox{.}(2018)]%
        {liu2018delayed}
\bibfield{author}{\bibinfo{person}{Lydia~T Liu}, \bibinfo{person}{Sarah Dean},
  \bibinfo{person}{Esther Rolf}, \bibinfo{person}{Max Simchowitz}, {and}
  \bibinfo{person}{Moritz Hardt}.} \bibinfo{year}{2018}\natexlab{}.
\newblock \showarticletitle{Delayed impact of fair machine learning}. In
  \bibinfo{booktitle}{\emph{International Conference on Machine Learning}}.
  PMLR, \bibinfo{pages}{3150--3158}.
\newblock


\bibitem[Liu et~al\mbox{.}(2020)]%
        {liu2020disparate}
\bibfield{author}{\bibinfo{person}{Lydia~T Liu}, \bibinfo{person}{Ashia
  Wilson}, \bibinfo{person}{Nika Haghtalab}, \bibinfo{person}{Adam~Tauman
  Kalai}, \bibinfo{person}{Christian Borgs}, {and} \bibinfo{person}{Jennifer
  Chayes}.} \bibinfo{year}{2020}\natexlab{}.
\newblock \showarticletitle{The disparate equilibria of algorithmic decision
  making when individuals invest rationally}. In
  \bibinfo{booktitle}{\emph{Proceedings of the 2020 Conference on Fairness,
  Accountability, and Transparency}}. \bibinfo{pages}{381--391}.
\newblock


\bibitem[Ma et~al\mbox{.}(2017)]%
        {ma2017runtime}
\bibfield{author}{\bibinfo{person}{Meiyi Ma}, \bibinfo{person}{John~A
  Stankovic}, {and} \bibinfo{person}{Lu Feng}.}
  \bibinfo{year}{2017}\natexlab{}.
\newblock \showarticletitle{Runtime monitoring of safety and performance
  requirements in smart cities}. In \bibinfo{booktitle}{\emph{Proceedings of
  the 1st ACM Workshop on the Internet of Safe Things}}.
  \bibinfo{pages}{44--50}.
\newblock


\bibitem[Mao and Chen(2012)]%
        {mao2012runtime}
\bibfield{author}{\bibinfo{person}{Jianhui Mao} {and} \bibinfo{person}{Liqian
  Chen}.} \bibinfo{year}{2012}\natexlab{}.
\newblock \showarticletitle{Runtime monitoring for cyber-physical systems: a
  case study of cooperative adaptive cruise control}. In
  \bibinfo{booktitle}{\emph{2012 Second International Conference on Intelligent
  System Design and Engineering Application}}. IEEE, \bibinfo{pages}{509--515}.
\newblock


\bibitem[Mehrabi et~al\mbox{.}(2021)]%
        {mehrabi2021survey}
\bibfield{author}{\bibinfo{person}{Ninareh Mehrabi}, \bibinfo{person}{Fred
  Morstatter}, \bibinfo{person}{Nripsuta Saxena}, \bibinfo{person}{Kristina
  Lerman}, {and} \bibinfo{person}{Aram Galstyan}.}
  \bibinfo{year}{2021}\natexlab{}.
\newblock \showarticletitle{A survey on bias and fairness in machine learning}.
\newblock \bibinfo{journal}{\emph{ACM Computing Surveys (CSUR)}}
  \bibinfo{volume}{54}, \bibinfo{number}{6} (\bibinfo{year}{2021}),
  \bibinfo{pages}{1--35}.
\newblock


\bibitem[Mouzannar et~al\mbox{.}(2019)]%
        {mouzannar2019fair}
\bibfield{author}{\bibinfo{person}{Hussein Mouzannar},
  \bibinfo{person}{Mesrob~I Ohannessian}, {and} \bibinfo{person}{Nathan
  Srebro}.} \bibinfo{year}{2019}\natexlab{}.
\newblock \showarticletitle{From fair decision making to social equality}. In
  \bibinfo{booktitle}{\emph{Proceedings of the Conference on Fairness,
  Accountability, and Transparency}}. \bibinfo{pages}{359--368}.
\newblock


\bibitem[Ogata et~al\mbox{.}(2010)]%
        {ogata2010modern}
\bibfield{author}{\bibinfo{person}{Katsuhiko Ogata} {et~al\mbox{.}}}
  \bibinfo{year}{2010}\natexlab{}.
\newblock \bibinfo{booktitle}{\emph{Modern control engineering}}.
  Vol.~\bibinfo{volume}{5}.
\newblock \bibinfo{publisher}{Prentice hall Upper Saddle River, NJ}.
\newblock


\bibitem[Podkopaev and Ramdas(2021)]%
        {podkopaev2021tracking}
\bibfield{author}{\bibinfo{person}{Aleksandr Podkopaev} {and}
  \bibinfo{person}{Aaditya Ramdas}.} \bibinfo{year}{2021}\natexlab{}.
\newblock \showarticletitle{Tracking the risk of a deployed model and detecting
  harmful distribution shifts}.
\newblock \bibinfo{journal}{\emph{arXiv preprint arXiv:2110.06177}}
  (\bibinfo{year}{2021}).
\newblock


\bibitem[Wainwright(2019)]%
        {wainwright2019high}
\bibfield{author}{\bibinfo{person}{Martin~J Wainwright}.}
  \bibinfo{year}{2019}\natexlab{}.
\newblock \bibinfo{booktitle}{\emph{High-dimensional statistics: A
  non-asymptotic viewpoint}}. Vol.~\bibinfo{volume}{48}.
\newblock \bibinfo{publisher}{Cambridge university press}.
\newblock


\bibitem[Waudby-Smith et~al\mbox{.}(2021)]%
        {waudby2021time}
\bibfield{author}{\bibinfo{person}{Ian Waudby-Smith}, \bibinfo{person}{David
  Arbour}, \bibinfo{person}{Ritwik Sinha}, \bibinfo{person}{Edward~H Kennedy},
  {and} \bibinfo{person}{Aaditya Ramdas}.} \bibinfo{year}{2021}\natexlab{}.
\newblock \showarticletitle{Time-uniform central limit theory, asymptotic
  confidence sequences, and anytime-valid causal inference}.
\newblock \bibinfo{journal}{\emph{arXiv preprint arXiv:2103.06476}}
  (\bibinfo{year}{2021}).
\newblock


\bibitem[Waudby-Smith and Ramdas(2020)]%
        {waudby2020estimating}
\bibfield{author}{\bibinfo{person}{Ian Waudby-Smith} {and}
  \bibinfo{person}{Aaditya Ramdas}.} \bibinfo{year}{2020}\natexlab{}.
\newblock \showarticletitle{Estimating means of bounded random variables by
  betting}.
\newblock \bibinfo{journal}{\emph{arXiv preprint arXiv:2010.09686}}
  (\bibinfo{year}{2020}).
\newblock


\bibitem[Zafar et~al\mbox{.}(2017)]%
        {zafar2017fairness}
\bibfield{author}{\bibinfo{person}{Muhammad~Bilal Zafar},
  \bibinfo{person}{Isabel Valera}, \bibinfo{person}{Manuel Gomez~Rodriguez},
  {and} \bibinfo{person}{Krishna~P Gummadi}.} \bibinfo{year}{2017}\natexlab{}.
\newblock \showarticletitle{Fairness beyond disparate treatment \& disparate
  impact: Learning classification without disparate mistreatment}. In
  \bibinfo{booktitle}{\emph{Proceedings of the 26th international conference on
  world wide web}}. \bibinfo{pages}{1171--1180}.
\newblock


\bibitem[Zafar et~al\mbox{.}(2019)]%
        {zafar2019fairness}
\bibfield{author}{\bibinfo{person}{Muhammad~Bilal Zafar},
  \bibinfo{person}{Isabel Valera}, \bibinfo{person}{Manuel Gomez-Rodriguez},
  {and} \bibinfo{person}{Krishna~P Gummadi}.} \bibinfo{year}{2019}\natexlab{}.
\newblock \showarticletitle{Fairness constraints: A flexible approach for fair
  classification}.
\newblock \bibinfo{journal}{\emph{The Journal of Machine Learning Research}}
  \bibinfo{volume}{20}, \bibinfo{number}{1} (\bibinfo{year}{2019}),
  \bibinfo{pages}{2737--2778}.
\newblock


\bibitem[Zemel et~al\mbox{.}(2013)]%
        {zemel2013learning}
\bibfield{author}{\bibinfo{person}{Rich Zemel}, \bibinfo{person}{Yu Wu},
  \bibinfo{person}{Kevin Swersky}, \bibinfo{person}{Toni Pitassi}, {and}
  \bibinfo{person}{Cynthia Dwork}.} \bibinfo{year}{2013}\natexlab{}.
\newblock \showarticletitle{Learning fair representations}. In
  \bibinfo{booktitle}{\emph{International conference on machine learning}}.
  PMLR, \bibinfo{pages}{325--333}.
\newblock


\bibitem[Zhang et~al\mbox{.}(2019)]%
        {zhang2019group}
\bibfield{author}{\bibinfo{person}{Xueru Zhang}, \bibinfo{person}{Mohammadmahdi
  Khaliligarekani}, \bibinfo{person}{Cem Tekin}, {et~al\mbox{.}}}
  \bibinfo{year}{2019}\natexlab{}.
\newblock \showarticletitle{Group retention when using machine learning in
  sequential decision making: the interplay between user dynamics and
  fairness}.
\newblock \bibinfo{journal}{\emph{Advances in neural information processing
  systems}}  \bibinfo{volume}{32} (\bibinfo{year}{2019}).
\newblock


\bibitem[Zhang and Liu(2021)]%
        {zhang2021fairness}
\bibfield{author}{\bibinfo{person}{Xueru Zhang} {and} \bibinfo{person}{Mingyan
  Liu}.} \bibinfo{year}{2021}\natexlab{}.
\newblock \showarticletitle{Fairness in learning-based sequential decision
  algorithms: A survey}.
\newblock In \bibinfo{booktitle}{\emph{Handbook of Reinforcement Learning and
  Control}}. \bibinfo{publisher}{Springer}, \bibinfo{pages}{525--555}.
\newblock


\end{thebibliography}

	\ifarxiv
		 \appendix
\newpage
\section{Proofs of the technical claims in Section~\ref{sec:interval estimator expected value}}
\label{sec:proofs_general_monitor}

For any sequence $\vec{x}=(x_s)_{s>0}$ and any $t \in \mathbb{N}$, let $\vec{x}_t$ be the subsequence of length $t$. For simplicity we assume that that $\vec{x}_0$ is the empty sequence.

\begin{definition}
	\label{definition:stochastic process}
	We call the stochastic process $\vec{O}=(O_t)_{t>0}=(X_t,Y_t)_{t>0}$ a stochastic process with observed linear shift,
    if for every time step $t$: $X_t$ is real valued;  $X_t-\expe(X_t\mid \vec{O}_{t-1})$ is a sub-exponential random variable with parameter $(\sigma^2,\nu)$ as defined in \cite[p.~26]{wainwright2019high}; $\expe_g(|X_t|\mid \seq{o}_{t-1})<\infty$ for every past observation sequence $\seq{o}_{t-1}$. And if there exists a function $\Delta$ such that for every concrete sequence $\vec{o}_t$
	\begin{align*}
		 \expe(X_{t+1} \mid \vec{o}_t) = \expe(X_{t} \mid \vec{o}_{t-1}) + \Delta(o_t).
	\end{align*}
\end{definition}
Moreover, for the remainder let $\vec{O}$ be a process as in Definition \ref{definition:stochastic process} and let $\vec{o}$ any concrete evaluation of $\vec{O}$. 

\begin{definition}
    \label{definition:estimator}
    We define the estimators such that for any time step $t$ and any sequence $\seq{o}_t$:
    \begin{align*}
        \hat{E}_1(\vec{o}_t) \coloneqq \frac{1}{t} \sum_{s=1}^{t}\left( x_s - \sum_{r=1}^{s-1} \Delta(o_r) \right) \quad \text{and} \quad \hat{E}(\vec{o}_t) \coloneqq \hat{E_1}(\vec{o}_t) + \sum_{r=1}^{s-1} \Delta(o_r).
    \end{align*}
\end{definition}
If clear from the context the we write $\hat{E}_1$ instead of $\hat{E}_1(\vec{o}_T)$. Analogue for $\hat{E}$.

The following lemma shows that $\hat{E}_1$ is an unbiased estimator of $\expe(X_1)$. We demonstrate this for $X_t$ and $Y_t$ being discrete. 
If $X_t$ or $Y_t$ are continuous, the sums over the respective domains have to be replaced with integrals and the probabilities have to be replaced with the probability density function. 
\begin{lemma}
    \label{lemma:estimator_unbiased}
    On the stochastic process $\vec{O}$ as in Definition \ref{definition:stochastic process} and for any time step $t$,
    the estimator $\hat{E}_1$ is an unbiased estimator for $\expe(X_1)$.
\end{lemma}
\begin{proof}
    To show this we need to establish that $\expe(\hat{E}_1(\vec{O}_t)) =\expe(X_1) $.
    First by linearity of expectation we obtain 
    \begin{align*}
        \expe(\hat{E}_1(\vec{O}_t))=  \frac{1}{t} \sum_{s=1}^{t}  \expe\left( X_s  - \sum_{r=1}^{s-1} \Delta(O_r) \right).
    \end{align*}
    Now we show for all $s$ that $\expe\left( X_s  - \sum_{r=1}^{s-1} \Delta(O_r) \right) =\expe(X_1)$.
    Notice that we need to take the expectation over the joint probability of the history. Hence, we have
    \begin{align*}
        \expe\left( X_s  - \sum_{r=1}^{s-1} \Delta(O_r) \right) = \sum_{z=(o_1, \dots, o_{s})} \pr\left(\vec{o}_{s}\right)\cdot  \left( x_{s} - \sum_{r=1}^{s-1} \Delta(o_r)\right).
    \end{align*}
    Using the definition of joint probability we obtain
    \begin{align*}
        \pr\left(\vec{o}_{s} \right) = \prod_{r=1}^{s-1} \pr\left( o_r \; \middle| \;\vec{o}_{r-1}\right) \cdot \pr\left(o_s \; \middle| \;\vec{o}_{s-1}\right) = \pr\left(\vec{o}_{s-1}\right)\cdot \pr\left(o_s\; \middle| \;\vec{o}_{s-1} \right)
    \end{align*}
    and thus we can rearrange the sum to
    \begin{align*}
        \sum_{z=(o_1,\dots, o_{s-1})} \left( \pr\left(\vec{o}_{s-1} \right) \cdot \left( \sum_{o_s} \pr\left(o_s \; \middle| \; \vec{o}_{s-1} \right) \cdot x_{s}  - \sum_{o_s}  \pr\left(o_s \; \middle| \; \vec{o}_{s-1} \right)  \cdot \sum_{r=1}^{s-1} \Delta(o_r)\right)\right).
    \end{align*}
    Notice that we can split the sum 
    \begin{align*}
        \sum_{o_s} \pr\left(o_s \; \middle| \; \vec{o}_{s-1} \right) \cdot x_{s} = \sum_{x_s} \sum_{y_s} \pr\left(x_s ,y_s \; \middle| \; \vec{o}_{s-1} \right) \cdot x_{s}.
    \end{align*}
    Since $x_s$ is constant when summing over $y_s$ and since $\sum_{r=1}^{s-1} \Delta(o_r)$ is constant when summing over $o_s$, the internal sum can be rewritten as 
    \begin{align*}
        \sum_{x_s}    x_s \cdot \sum_{y_s} \pr\left(x_s ,y_s \; \middle| \; \vec{o}_{s-1} \right)  - \sum_{r=1}^{s-1} \Delta(o_r) \cdot \sum_{o_s} \pr\left(o_s \; \middle| \; \vec{o}_{s-1}  \right).
    \end{align*}
    By the law of total probability we obtain 
    \begin{align*}
        \expe\left( X_s  - \sum_{r=1}^{s-1} \Delta(O_s) \right) 
        &= \sum_{z=(o_1,\dots, o_{s-1})} \pr\left( \vec{o}_{s-1}\right)   \cdot \left( \sum_{x_s}  x_s \cdot\pr\left( x_s \; \middle| \;\vec{o}_{s-1} \right) - \sum_{r=1}^{s-1} \Delta(o_r) \right) \\
        &=
        \sum_{z=(o_1,\dots, o_{s-1})}  \pr\left( \vec{o}_{s-1}\right)  \cdot \left(\expe\left( X_s \; \middle| \; \vec{o}_{s-1}\right) - \sum_{r=1}^{s-1} \Delta(o_r) \right).
    \end{align*}
    Moreover, we know that $\expe\left( X_s \; \middle| \; \vec{o}_{s-1}\right) =  \expe\left( X_{s-1} \; \middle| \; \vec{o}_{s-2}\right) + \Delta(o_{s-2})$. By repeated application of this equality we obtain  
    \begin{align*}
        \expe\left( X_s  - \sum_{r=1}^{s-1} \Delta(O_s) \right)  &=
        \sum_{z=(o_1,\dots, o_{s-1})}\pr\left( \vec{o}_{s-1}\right) \cdot \left(\expe\left( X_1\right) + \sum_{r=1}^{s-1} \Delta(o_r)  - \sum_{r=1}^{s-1} \Delta(o_r) \right)  \\
        &= \expe(X_1) \cdot \sum_{z=(o_1,\dots, o_{s-1})} \pr\left( \vec{o}_{s-1}\right) = \expe(X_1).
    \end{align*}
\end{proof}

The following supporting lemma demonstrates a useful equality.
\begin{lemma}
    \label{lemma:martingale_conditional_expectation_equalities}
    For the stochastic process $\vec{O}$ as in Definition \ref{definition:stochastic process} and for any $t,k \in \mathbb{N}$ such that $k<t$, the following equalities hold
    \begin{align*}
        \expe\left( X_t \; \middle| \; \vec{O}_{k}\right) 
        = 
        \expe\left( X_1 \right) + \sum_{s=1}^{t-1} \expe\left(  \Delta(O_s)\; \middle| \; \vec{O}_{k}\right) 
        =
        \expe\left( X_1 \right) + \sum_{s=1}^{k} \Delta(O_s) + \sum_{s=k+1}^{t-1} \expe\left(  \Delta(O_s)\; \middle| \; \vec{O}_{k}\right).
    \end{align*}
\end{lemma}
\begin{proof}
    First let $t=k+1$ then by the properties of the stochastic process.
    \begin{align*}
        \expe\left( X_{k+1} \; \middle| \; \vec{O}_{k} \right) = \expe\left( X_{k} \; \middle| \; \vec{O}_{k-1} \right) + \Delta(O_k).
    \end{align*}
    Hence, by repeated substitution we obtain 
    \begin{align*}
        \expe\left( X_{k+1} \; \middle| \; \vec{O}_{k} \right) = \expe\left( X_1 \right) + \sum_{s=1}^{k} \Delta(O_s).
    \end{align*}
    Let $t>k+1$ and notice that
    \begin{align*}
        \expe\left( X_{t} \; \middle| \; \vec{O}_{k} \right) = \expe \left( \expe\left( X_{t} \; \middle| \;   \vec{O}_{k+1 }\right) \; \middle| \;  \vec{O}_{k}  \right),
    \end{align*}
    which follows from the definition of conditional expectation.
    Hence, by repeated application of this equality we obtain the following nested expectation expression. 
    \begin{align*}
        \expe\left( X_{t} \; \middle| \; \vec{O}_{k} \right) = \expe \left( \dots \expe\left(\expe\left(\expe\left( X_{t} \; \middle| \;   \vec{O}_{t-1 }\right) \; \middle| \; \vec{O}_{t-2 }\right) \; \middle| \; \vec{O}_{t-3 }\right) \dots  \; \middle| \;  \vec{O}_{k}  \right).
    \end{align*}
    First we notice that we can decompose the innermost expectation using Lemma \ref{lemma:martingale_conditional_expectation_equalities} to obtain 
    \begin{align*}
        \expe\left( X_{t} \; \middle| \; \vec{O}_{k} \right) = \expe \left( \dots \expe\left(\expe\left( \expe\left( X_1 \right) + \sum_{s=1}^{t-1} \Delta(O_s)  \; \middle| \; \vec{O}_{t-2 }\right) \; \middle| \; \vec{O}_{t-3 }\right) \dots  \; \middle| \;  \vec{O}_{k}  \right).
    \end{align*}
    Now we apply the tower property of the conditional expectation and linearity of expectation to obtain
    \begin{align*}
        \expe\left( X_{t} \; \middle| \; \vec{O}_{k} \right) &= \expe \left(\expe\left( X_1 \right) + \sum_{s=1}^{t-1} \Delta(O_s)    \; \middle| \;  \vec{O}_{k}  \right)  \\
        &= \expe \left(\expe\left( X_1 \right)  \; \middle| \;  \vec{O}_{k}  \right)  + \sum_{s=1}^{t} \expe\left(\Delta(O_s)    \; \middle| \;  \vec{O}_{k}  \right) \\
        &= \expe\left( X_1 \right) + \sum_{s=1}^k \Delta(O_s) +  \sum_{s=k+1}^{t-1} \expe\left(  \Delta(O_s)\; \middle| \; \vec{O}_{k}\right).
    \end{align*}
    The last step is possible because the expected value is a constant and because $\expe\left( f(W)\cdot V \; \middle| \; W\right) = f(W) \cdot \expe\left(  V \; \middle| \; W\right)$ for any random variables $V$ and $W$.

\end{proof}

The following lemma shows that the estimator $\hat{E}$ estimates the quantity we are interested in, i.e.\ $\expe\left(X_t \; \middle| \; \vec{o}_{t-1} \right)$.

\begin{lemma}
    For the stochastic process $\vec{O}$ as in Definition \ref{definition:stochastic process} and for any time step $t$:
    \begin{align*}
        \expe\left(X_t \;\middle|\; \vec{O}_{t-1}\right) = \expe\left(\hat{E}\left(\vec{O}_t\right)\;\middle|\; \vec{O}_{t-1} \right).
    \end{align*}
\end{lemma}
\begin{proof}
    Here we utilize the properties of conditional expectation and Lemma \ref{lemma:martingale_conditional_expectation_equalities} to show that
    \begin{align*}
        \expe\left(X_t \;\middle|\; \vec{O}_{t-1}\right) 
        = 
        \expe(X_1)+ \sum_{s=1}^{t-1} \Delta(O_s) 
        = 
        \expe\left(\hat{E}_1\left(\vec{O}_t\right)\right) + \sum_{s=1}^{t-1} \Delta(O_s) 
        =
        \expe\left(\hat{E}\left(\vec{O}_t\right)\;\middle|\; \vec{O}_{t-1} \right).
    \end{align*}
\end{proof}

The following lemma provides insight into the martingale difference sequence used for the approximation of our estimator. 
\begin{lemma}
    \label{lemma:estimator_difference}
    For the stochastic process $\vec{O}$ as in Definition \ref{definition:stochastic process} and for any time step $t$ the sequence 
    \begin{align*}
        \left( \expe\left( \hat{E}_1(\vec{O}_t) \; \middle| \; \vec{O}_s \right) \right)_{s\in [1\twodots t]}
    \end{align*}
     is a martingale such that for all $s\in[1\twodots t-1]$:
    \begin{align*}
        \expe\left( \hat{E}_1(\vec{O}_t) \; \middle| \; \vec{O}_{s+1} \right)- \expe\left( \hat{E}_1(\vec{O}_t) \; \middle| \; \vec{O}_s \right) =\frac{1}{t}\left(X_{s+1} - \expe\left( X_{s+1} \; \middle| \; \vec{O}_s\right)\right).
    \end{align*}
\end{lemma}
\begin{proof}
    First, $M$ is a (Doob-)martingale, which follows from the definition of conditional expectation, i.e.\ for any $k\in[1\twodots t-1]$
    \begin{align*}
    \expe\left(\expe\left(\hat{E}_1(\vec{O}_t) \; \middle| \; \vec{O}_{k+1} \right) \; \middle| \;  \vec{O}_{k} \right)= \expe\left(\hat{E}_1(\vec{O}_t)\; \middle| \; \vec{O}_k \right),
    \end{align*}
    and the fact that the conditional expected value of $|X_t|$ is bounded, as required in Definition \ref{definition:stochastic process}.
    Second, we need to compute for any $k\in[1\twodots t-1]$ the difference between 
    \begin{align*}
        \expe\left(\frac{1}{t} \sum_{s=1}^{t} X_s - \sum_{r=1}^{s-1} \Delta(O_r)  \; \middle| \; \vec{O}_{k+1} \right) - \expe\left( \frac{1}{t} \sum_{s=1}^{t} X_s - \sum_{r=1}^{s-1} \Delta(O_r) \; \middle| \; \vec{O}_k \right).
    \end{align*}
    First notice that from Lemma \ref{lemma:martingale_conditional_expectation_equalities} if $s>k$,
    \begin{align*}
        \expe\left( X_s \; \middle| \; \vec{O}_k \right) - \sum_{r=1}^{s-1} \expe\left(\Delta(O_r) \; \middle| \; \vec{O}_k \right) 
        =\expe\left( X_1\right)
    \end{align*}
    and for $s\leq k$,
    \begin{align*}
        \expe\left( X_s \; \middle| \; \vec{O}_k \right) - \sum_{r=1}^{s-1} \expe\left(\Delta(O_r) \; \middle| \; \vec{O}_k  \right) = X_s + \sum_{r=1}^{s-1} \Delta(O_r).
    \end{align*}
    Therefore, now we split the sums at $k$ and $k+1$. Then we apply the previous observations to obtain
    \begin{align*}
        &\expe\left( \sum_{s=1}^{t} X_s - \sum_{r=1}^{s-1} \Delta(O_r)  \; \middle| \; \vec{O}_{k+1} \right) - \expe\left( \sum_{s=1}^{t} X_s - \sum_{r=1}^{s-1} \Delta(O_r) \; \middle| \; \vec{O}_k \right) \\
        &= \sum_{s=1}^{k+1}\left( X_s + \sum_{r=1}^{s-1} \Delta(O_r)\right) + (t-k-1)\cdot \expe\left( X_1\right) - \sum_{s=1}^{k} \left(X_s + \sum_{r=1}^{s-1} \Delta(O_r)\right) - (t-k)\cdot \expe\left( X_1\right)  \\
        &= X_{k+1} + \sum_{s=1}^{k}  \Delta(O_s) - \expe(X_1)  
        = X_{k+1} - \expe\left( X_{k+1} \; \middle| \; \vec{O}_k\right),
    \end{align*}
    which concludes the proof.
\end{proof}
With the martingale difference sequence bounded, we can use an Azuma-style concentration inequality to demonstrate that our estimator $\hat{E}_1$ concentrates well around its mean, and thus $\expe(X_1)$.

\begin{lemma}
    \label{lemma:martingale_estimator_concentration_general}
    
	Let $\vec{O}$ be a stochastic process as in Definition \ref{definition:stochastic process}. Let $\varepsilon>0$.
	For every time step $t$,
    \begin{align*}
        \pr\left(\left| \hat{E}_1(\vec{O}_t) - \expe(X_1)\right| \geq \varepsilon \right) \leq 2\exp\left( - \min\left\lbrace \frac{t\cdot \varepsilon^2}{2 \sigma^2}, \frac{t\cdot \varepsilon}{2\nu} \right\rbrace \right)
    \end{align*}
    and $\hat{E}_1(\vec{O}_t) - \expe(X_1)$ is sub-exponential with parameters $(\frac{\sigma^2}{t},\frac{\nu}{t})$.
\end{lemma}
\begin{proof}
    Recall the Doob martingale as defined in Lemma \ref{lemma:estimator_difference}. For any $s\in[1\twodots t-1]$, the corresponding martingale difference is $D_s\coloneqq \frac{1}{t}\cdot \left( X_{s+1} - \expe\left( X_{s+1} \; \middle| \; \vec{O}_s\right) \right)$. That is, it is a sub-exponential random variable with parameters $(\sigma^2,\nu)$ divided by $t$. Hence, from the definition of a sub-exponential random variables it follows directly that $D_s$ is a sub-exponential random variable with parameters $(\frac{\sigma^2}{t^2},\frac{\nu}{t})$. 
    By applying a standard concentration bound for sub-exponential martingale difference sequences we obtain 
    \begin{align*}
        \pr\left(\left| \hat{E}_1(\vec{O}_t) - \expe(\hat{E}_1(\vec{O}_t)) \right| \geq \varepsilon \right) \leq 2\exp\left( - \min\left\lbrace \frac{t\cdot \varepsilon^2}{2 \sigma^2}, \frac{t\cdot \varepsilon}{2\nu} \right\rbrace \right).
    \end{align*}
    as well as the knowledge that $\hat{E}_1(\vec{O}_t) - \expe(\hat{E}_1(\vec{O}_t))$ is sub-exponential with parameters $(\frac{\sigma^2}{t},\frac{\nu}{t})$. 
    By applying Lemma \ref{lemma:estimator_unbiased} we can replace $\expe(\hat{E}(\vec{O}_t))$ with $\expe(X_1)$, thus concluding the proof.
\end{proof}

Now we can show the soundness of our monitor, by transferring this bound to $\hat{E}$.
\begin{theorem}
    \label{theorem:martingale_estimator_concentration_general_at_T}
    Let $\delta\in [0,1]$.
	Let $\vec{O}$ be a stochastic process as defined in Definition \ref{definition:stochastic process}.
	For every time step $t$,
	\begin{align*}
        \left| \expe(X_t\mid \vec{O}_{t-1})-\hat{E}(\vec{O}_t) \right| \leq\max\left\lbrace 
		\sqrt{\frac{2\sigma^2}{t} \cdot \log\left( \frac{2}{\delta}\right)} ,
		\frac{2\nu}{t}\log\left( \frac{2}{\delta}\right) \right\rbrace
    \end{align*}
	holds with probability $1-\delta$. 
\end{theorem}
\begin{proof}
    We start by noticing that
    \begin{align*}
        \pr\left(\left| \expe\left( \hat{E}_1(\vec{O}_t)\right) - \expe(X_1) + \sum_{s=1}^{t-1} \left( \Delta(O_s)-\Delta(O_s) \right)\right| \geq \varepsilon \right) &= \pr\left(\left| \expe\left( \hat{E}_1(\vec{O}_t)\right) + \sum_{s=1}^{t-1} \Delta(O_s) - \expe(X_1)- \sum_{s=1}^{t-1} \Delta(O_s)\right| \geq \varepsilon \right)
        \\
        & = \pr\left(\left| \expe\left( \hat{E}(\vec{O}_t)\right)  - \expe\left( X_t\; \middle|\; \vec{O}_{t-1}  \right)\right| \geq \varepsilon \right),
    \end{align*}
    which we can bound by applying Lemma \ref{lemma:martingale_estimator_concentration_general}. To conclude we express $\varepsilon$ as function of $\delta$.
\end{proof}

\section{Proofs of the technical claims in Section~\ref{sec:attention}}

To show that the monitor for the attention setting is sound. We need to establish that for $X \sim \poisson{\lambda}$, $X-\lambda$ is sub-exponential. 
\begin{lemma}
    \label{lemma:subexponential_poisson}
    Let $X \sim \poisson{\lambda}$ then $X-\lambda$ is a sub-exponential random variable with parameter $(2\lambda,2)$.
\end{lemma}
\begin{proof}
    Using the definition of sub-exponential random variables, it suffices to check whether the moment generating function can be bounded on the appropriate interval using those parameters.
    First we observe that for any $c\in\mathbb{R}$
    \begin{align*}
        \expe\left( e^{c (X-\lambda)}\right)
        = \sum_{k=0}^{\infty} e^{c\cdot k}\cdot e^{-c\lambda}\cdot e^{-\lambda}\cdot \frac{\lambda^k}{k!} = e^{-c\lambda}\cdot e^{-\lambda}\cdot e^{e^c\lambda} = e^{\lambda\cdot (e^c-c-1)}.
    \end{align*}
    In the following, we establish that for  $|c|\leq \frac{1}{2}$ the following holds
    \begin{align*}
        e^{\lambda\cdot (e^c-c-1)} \leq e^{c^2\cdot \lambda}, 
    \end{align*}
   	or, equivalently, the following holds
    \begin{align}
        e^c-c-1 \leq c^2. \label{eq:proof:subexponential_poisson:a} 
    \end{align}
    To establish \eqref{eq:proof:subexponential_poisson:a}, define the function $f(c) \coloneqq e^c-c^2-c-1$, and note that it is continuous everywhere.
    The function $f(c)$ attains its extrema at the roots of the following equation:
    \begin{align*}
    		f'(c) = e^c-2c -1 =0.
    \end{align*}
    Now observe that the roots of the above equation are exactly the points where $e^c$ intersects with the line $2c+1$.
    Since $e^c$ is a convex function (follows from the fact that its double derivative is positive everywhere), hence it can have at most two intersection points with $2c+1$, and they are $c=0$ and $c\approx 1.256$.
    Thus, clearly there is only one extremum point of $f$ within the interval $c\in [-0.5,0.5]$, which is at $c=0$.
    Since $f''(0) = e^0-2 = -1$, hence $c=0$ is a maxima.
   	Lastly, by observing that $f(0) = 0$, \eqref{eq:proof:subexponential_poisson:a} follows.
\end{proof}

The following theorem demonstrates that the expected value of the attention allocation property has a closed form. 

\begin{theorem} \label{theorem:attention_property_closed_form}
        Let  $X \sim \poisson{\lambda}$ and $y$ be a constant. Let $f(X)=\frac{\hat{X}}{X+1}$, where $\hat{X}=\min \set{y,X+1}$. Then $\expe(f(X))=\eta(y,\lambda)$, where $\eta$ is the function defined in eq.~\eqref{eq:eta}.
    \end{theorem}
    
    \begin{proof}
        \begin{equation*}
            \expe(f(X)) = \sum_{k=0}^\infty f(k) \cdot \pr(X=k)
        \end{equation*}
        Observe that $\hat{X}=X+1$ and $f(X)=1$ if $X<y$. Thus we have
        \begin{align*}
            \expe(f(X)) & = \sum_{k=0}^{y-1} 1 \cdot \pr(X=k) + \sum_{k=y}^\infty \frac{y}{k+1} \cdot \pr(X=k) \\ 
            & \underbrace{e^{-\lambda} \sum_{k=0}^{y-1} \frac{\lambda^k}{k!}}_{A} + \underbrace{a \cdot e^{-\lambda} \sum_{k=y}^\infty \frac{\lambda^k}{(k+1)\cdot k!}}_B
        \end{align*}
    
        We simplify $B$ as follows:
        \begin{align*}
            B=a \cdot e^{-\lambda} \sum_{k=y}^\infty \frac{\lambda^k}{(k+1)\cdot k!} & = \frac{y \cdot e^{-\lambda}}{\lambda} \sum_{k=y}^{\infty} \frac{\lambda^{k+1}}{(k+1)!} \\
            & = \frac{y \cdot e^{-\lambda}}{\lambda} \left( \sum_{k=0}^\infty \frac{\lambda^{k+1}}{(k+1)!} - \sum_{k=0}^{a-1} \frac{\lambda^{k+1}}{(k+1)!} \right) \\
            & = \frac{y \cdot e^{-\lambda}}{\lambda} \left( \left(e^\lambda - 1\right) - \sum_{k=0}^{y-1} \frac{\lambda^{k+1}}{(k+1)!} \right) \\
            & = \frac{y}{\lambda} - \frac{y \cdot e^{-\lambda}}{\lambda} \left(1 + \sum_{k=0}^{y-1} \frac{\lambda^{k+1}}{(k+1)!} \right)
        \end{align*}
        Hence, we obtain
        \begin{align*}
            \expe(f(X))=A+B & = e^{-\lambda} \sum_{k=0}^{y-1} \frac{\lambda^k}{k!} + \frac{y}{\lambda} - \frac{y \cdot e^{-\lambda}}{\lambda} \left(1 + \sum_{k=0}^{y-1} \frac{\lambda^{k+1}}{(k+1)!} \right) \\
            & = e^{-\lambda} \sum_{k=0}^{y-1} \frac{\lambda^k}{k!} \left(1 - \frac{y}{k+1}\right) + \frac{y}{\lambda}\left(1 - e^{-\lambda}\right)
        \end{align*}
    \end{proof}
    
    \begin{theorem} \label{theorem:eta_decreasing}
        Given a natural number $y>0$, function $\eta(y, \lambda)$ as defined in eq.~\eqref{eq:eta} is strictly decreasing for $\lambda > 0$.
    \end{theorem}
    
    \begin{proof}
        Assume $0<\lambda_1<\lambda_2$. From Lemmas \ref{lemma:eta_sum_decreasing} and \ref{lemma:eta_term_decreasing}, we have
        \begin{align*}
            e^{-\lambda_1} \sum_{k=0}^{y-1} \frac{\lambda_1^k}{k!} \left( 1 - \frac{y}{k+1} \right) & > e^{-\lambda_2} \sum_{k=0}^{y-1} \frac{\lambda_2^k}{k!} \left( 1 - \frac{y}{k+1} \right) \\
            \frac{y}{\lambda_1}\left(1 - e^{-\lambda_1}\right) & > \frac{y}{\lambda_2}\left(1 - e^{-\lambda_2}\right)
        \end{align*}
        Thus we have
        \begin{equation*}
            \eta(y, \lambda_1) = e^{-\lambda_1} \sum_{k=0}^{y-1} \frac{\lambda_1^k}{k!} \left( 1 - \frac{y}{k+1} \right) + \frac{y}{\lambda_1}\left(1 - e^{-\lambda_1}\right) > e^{-\lambda_2} \sum_{k=0}^{y-1} \frac{\lambda_2^k}{k!} \left( 1 - \frac{y}{k+1} \right) + \frac{y}{\lambda_2}\left(1 - e^{-\lambda_2}\right) = \eta(y, \lambda_2)
        \end{equation*}
    \end{proof}

    The following lemma is a supporting lemma.

    \begin{lemma} \label{lemma:eta_sum_decreasing}
        Given a natural number $y>0$ and real numbers $0 < \lambda_1 < \lambda_2$,
        \begin{equation*}
            e^{-\lambda_1} \sum_{k=0}^{y-1} \frac{\lambda_1^k}{k!} \left( 1 - \frac{y}{k+1} \right) > e^{-\lambda_2} \sum_{k=0}^{y-1} \frac{\lambda_2^k}{k!} \left( 1 - \frac{y}{k+1} \right) 
        \end{equation*}
    \end{lemma}
    \begin{proof}
        \begin{equation*}
            0 < \lambda_1 < \lambda_2 \implies \forall k \in \Naturals.~\lambda_1^k \leq \lambda_2^k \implies \forall k \in \Naturals.~ \frac{\lambda_1^k}{k!} \leq \frac{\lambda_2^k}{k!}
        \end{equation*}
        For any $k \in \set{0,1,\cdots,y-1}$, we have $\frac{y}{k+1} \geq 1$. Thus we have
        \begin{align*}
            \frac{\lambda_1^k}{k!} \left(1 - \frac{y}{k+1}\right) & \geq \frac{\lambda_2^k}{k!} \left(1 - \frac{y}{k+1}\right) \\
            \implies \sum_{k=0}^{y-1}\frac{\lambda_1^k}{k!} \left(1 - \frac{y}{k+1}\right) & \geq  \sum_{k=0}^{y-1} \frac{\lambda_2^k}{k!} \left(1 - \frac{y}{k+1}\right)
        \end{align*}
        For any $0 < x_1 < x_2$, we have $e^{-x_1} > e^{-x_2} > 0$. Thus we have
        \begin{equation*}
            e^{-\lambda_1} \sum_{k=0}^{y-1}\frac{\lambda_1^k}{k!} \left(1 - \frac{y}{k+1}\right) > e^{-\lambda_2} \sum_{k=0}^{y-1} \frac{\lambda_2^k}{k!} \left(1 - \frac{y}{k+1}\right)
        \end{equation*}
    \end{proof}
    
    The following lemma demonstrates that $\eta$ is strictly decreasing for $\lambda$. Hence, our monitor can efficiently compute its estimate.

    \begin{lemma} \label{lemma:eta_term_decreasing}
        Given a natural number $y>0$, $\frac{y}{\lambda}\left(1 - e^{-\lambda}\right)$ is strictly decreasing for $\lambda > 0$.
    \end{lemma}
    
    \begin{proof}
        Let $f(\lambda)=\frac{y}{\lambda}\left(1 - e^{-\lambda}\right)$. We have
        \begin{align*}
            f'(\lambda) & = \frac{-y}{\lambda^2} \left(1 - e^{-\lambda}\right) + e^{-\lambda}\frac{y}{\lambda} \\
            & = \frac{y}{\lambda} \left(\frac{e^{-\lambda} - 1}{\lambda} + e^{-\lambda}\right)
        \end{align*}
    
        We show that for any $\lambda>0$, $1+\lambda < e^\lambda$. Let $g_1(\lambda)=1+\lambda$ and $g_2(\lambda)=e^\lambda$. We have $g_1(0)=g_2(0)=0$. Let $h(\lambda)=g_1(\lambda)-g_2(\lambda)$. We have $h'(\lambda)=1-e^{\lambda}<0$ for any $\lambda > 0$; hence, $h$ is strictly decreasing for any $\lambda > 0$, and we obtain $1+\lambda < e^\lambda$ for any $\lambda > 0$. From this result, we have
        \begin{align*}
            & \frac{1+\lambda}{e^\lambda} = e^{-\lambda}(1+\lambda) < 1 \\
            \implies & e^{-\lambda}(1+\lambda) - 1  = e^{-\lambda}-1 + \lambda e^{-\lambda} < 0 \\
            \implies & \frac{e^{-\lambda}-1}{\lambda} + e^{-\lambda} < 0 \\
            \implies & \frac{y}{\lambda} \left(\frac{e^{-\lambda}-1}{\lambda} + e^{-\lambda}\right) < 0
        \end{align*}
    
        Thus, $f(\lambda)$ is strictly decreasing for any $\lambda > 0$.
    \end{proof}

	\fi
\end{document}